\documentclass[12pt]{article}

\usepackage[T1]{fontenc}
\usepackage{lmodern}

\usepackage{setspace}
\usepackage[margin=1.25in]{geometry}

\usepackage[dvipsnames]{xcolor}
\usepackage{pdfpages}
\usepackage{float}
\usepackage{multirow}
\usepackage{caption}
\usepackage{subcaption}
\usepackage{epigraph}

\AtBeginDocument{%
  \everymath{\color{Violet}}%
  \everydisplay{\color{Violet}}%
}

\usepackage{mathtools}
\usepackage{amssymb}
\usepackage{amsthm}
\usepackage{amsfonts}
\usepackage{aliascnt}
\usepackage{accents}
\usepackage{dutchcal}
\usepackage{bbm}

\usepackage{enumitem}
\usepackage{sgame}
\usepackage{tikz}
\usepackage{tikz-cd}
\usetikzlibrary{calc,shapes,arrows}
\usepackage{tcolorbox}
\usepackage{listings}
\usepackage[normalem]{ulem}

\usepackage[round]{natbib}
\newcommand{\1}{\mathbbm{1}}

\newtheorem{theorem}{Theorem}[section]

\newaliascnt{proposition}{theorem}
\newtheorem{proposition}[proposition]{Proposition}
\aliascntresetthe{proposition}

\newaliascnt{lemma}{theorem}
\newtheorem{lemma}[lemma]{Lemma}
\aliascntresetthe{lemma}

\newaliascnt{corollary}{theorem}
\newtheorem{corollary}[corollary]{Corollary}
\aliascntresetthe{corollary}

\newaliascnt{claim}{theorem}
\newtheorem{claim}[claim]{Claim}
\aliascntresetthe{claim}

\theoremstyle{definition}

\newaliascnt{definition}{theorem}
\newtheorem{definition}[definition]{Definition}
\aliascntresetthe{definition}

\newaliascnt{example}{theorem}

\aliascntresetthe{example}

\newaliascnt{assumption}{theorem}
\newtheorem{assumption}[assumption]{Assumption}
\aliascntresetthe{assumption}

\newaliascnt{condition}{theorem}

\aliascntresetthe{condition}

\newaliascnt{question}{theorem}

\aliascntresetthe{question}

\newaliascnt{remark}{theorem}

\aliascntresetthe{remark}

\newaliascnt{remarks}{theorem}

\aliascntresetthe{remarks}

\newaliascnt{aside}{theorem}

\aliascntresetthe{aside}

\newaliascnt{note}{theorem}

\aliascntresetthe{note}

\usepackage{thmtools}
\usepackage{thm-restate}

\usepackage{hyperref}

\hypersetup{
    colorlinks=true,
    linkcolor=OrangeRed,
    filecolor=Thistle,
    urlcolor=Thistle,
    citecolor=Thistle,
}

\usepackage[nameinlink]{cleveref}

\crefname{theorem}{theorem}{theorems}
\Crefname{theorem}{Theorem}{Theorems}
\crefname{proposition}{proposition}{propositions}
\Crefname{proposition}{Proposition}{Propositions}
\crefname{lemma}{lemma}{lemmas}
\Crefname{lemma}{Lemma}{Lemmas}
\crefname{corollary}{corollary}{corollaries}
\Crefname{corollary}{Corollary}{Corollaries}
\crefname{claim}{claim}{claims}
\Crefname{claim}{Claim}{Claims}
\crefname{definition}{definition}{definitions}
\Crefname{definition}{Definition}{Definitions}
\crefname{example}{example}{examples}
\Crefname{example}{Example}{Examples}
\crefname{assumption}{assumption}{assumptions}
\Crefname{assumption}{Assumption}{Assumptions}

\makeatletter
\let\cref@old@isrefconsecutive\cref@isrefconsecutive
\def\cref@isrefconsecutive#1#2{%
  \begingroup
    \def\cref@assumptiontype{assumption}%
    \cref@gettype{#1}{\cref@typea}%
    \ifx\cref@typea\cref@assumptiontype
      \endgroup
      \@cref@refconsecutivefalse
    \else
      \endgroup
      \cref@old@isrefconsecutive{#1}{#2}%
    \fi
}
\makeatother

\crefname{condition}{condition}{conditions}
\Crefname{condition}{Condition}{Conditions}
\crefname{question}{question}{questions}
\Crefname{question}{Question}{Questions}
\crefname{remark}{remark}{remarks}
\Crefname{remark}{Remark}{Remarks}
\crefname{remarks}{remarks}{remarks}
\Crefname{remarks}{Remarks}{Remarks}
\crefname{aside}{aside}{asides}
\Crefname{aside}{Aside}{Asides}
\crefname{note}{note}{notes}
\Crefname{note}{Note}{Notes}
\crefname{appendix}{appendix}{appendices}
\Crefname{appendix}{Appendix}{Appendices}

\newcommand{\secref}[1]{\hyperref[#1]{\S\ref*{#1}}}

\definecolor{backcolour}{rgb}{0.63,0.79,0.95}
\lstdefinestyle{mystyle}{
  backgroundcolor=\color{backcolour},
  basicstyle=\ttfamily\footnotesize,
  breakatwhitespace=false,
  breaklines=true,
  captionpos=b,
  keepspaces=true,
  numbers=left,
  numbersep=5pt,
  showspaces=false,
  showstringspaces=false,
  showtabs=false,
  tabsize=2
}
\begin{document} 
\title{Bins}
\author{Mark Whitmeyer\thanks{Arizona State University. Email: \href{mailto:mark.whitmeyer@gmail.com}{mark.whitmeyer@gmail.com}.  For KS.}}
\date{\today}
\maketitle

\begin{abstract}
I characterize the undominated ways of transmitting, storing, and organizing information under capacity constraints, showing that they take a precise geometric form. Under a regularity assumption on the information to be organized, they are convex subdivisions of the space of beliefs constructed by recursively subdividing cells according to weighted proximity criteria.
\end{abstract}

The efficient transmission, storage, and organization of information under constraints is a fundamental line of study, precisely the purview of \textbf{information theory}. Suppose, for example, that a sensor measures temperature precisely but has only  four bins in which to group them. An agent then uses a representative temperature for each bin to approximate the original measurement. Information theory calls this operation \emph{quantization}, and the classical problem is to choose the bins and their representatives to minimize the resulting approximation error.\footnote{Measured, e.g., by average squared error \citep{lloyd1982least,grayNeuhoff1998quantization}. \citet{shannon1959coding} develops the broader tradeoff between information rate and fidelity.} I study the corresponding economic question: \textbf{which finite arrangements of information make efficient use of a limited vocabulary when information is valuable for the decisions it informs?}

Mirroring the classical exercise, I take as given a source of information and ask how to efficiently summarize it using a limited number of messages. Whereas the quantization problems above minimize a specified measure of approximation error, I use a dominance criterion in the spirit of \citet{blackwell1953equivalent}. Specifically, I fix a number of messages \(k\) and characterize the \textit{summaries}--possibly randomized reports using \(k\) or fewer messages--for which no other such summary makes the recipient no worse off in every decision problem and strictly better off in some.

A record containing \(b\) binary digits, \textit{viz., bits}, can distinguish at most \(2^b\) possibilities, which lends the bound on the number of messages a technological interpretation. The capacity constraint can also reflect the cost of supporting a shared vocabulary. A label communicates information only when its users understand its meaning, and adding a category requires people to learn, remember, and consistently apply another distinction. Even an infrequently used category must have a defined meaning if the system is to use it.\footnote{\citet{cremerGaricanoPrat2007language} model organizational communication through a language whose vocabulary is limited by the number of words agents can learn. For evidence that color vocabularies convey information efficiently given their number of terms, refer to \citet{regierKempKay2015communication}.}

The same consideration applies to the information people retain for their own use. When someone groups experiences into categories, she preserves distinctions across categories while discarding some distinctions within them. A smaller repertoire can economize on storing descriptions, retrieving relevant experiences, and interpreting new observations.\footnote{Cf., \citet{hellmanCover1970learning,rosch1978principles,fryerJackson2008categorical,wilson2014bounded}. \citet{reed1972pattern} emphasizes the memory savings from representing categories by their prototypes. \citet{bordaloGennaioliLanzaniShleifer2026cognitive} model how categorizing a new decision problem using past experience shapes attention and choice. On how selective recall shapes valuation and probability judgments, see \citet{bordaloGennaioliShleifer2020memory,bordaloEtAl2023memory}.} The number of messages also measures one aspect of the complexity of acting on information. Each additional report can create another case for a contingent rule or plan to distinguish.\footnote{\citet{gulPesendorferStrzalecki2017coarse} study contingent consumption plans with a bounded number of distinct consumption levels. \citet{abreuRubinstein1988structure} measure the complexity of a strategy by the number of states in its implementing automaton. \citet{oprea2020rule} finds experimentally that adding states increases the subjective cost of implementing a rule.} The number of possible monetary outcomes can also affect how people evaluate lotteries and their willingness to accept them.\footnote{\citet{huckWeizsacker1999risk,moffattSitziaZizzo2015heterogeneity,oberholzerOlschewskiScheibehenne2024complexity,puri2025simplicity} provide experimental evidence. \citet{fudenbergPuri2022simplicity,hu2023procedural} provide foundations for coarse evaluation of lotteries.}

This motivation of my exercise is supported experimentally. Indeed, subjects can classify the same stimuli less accurately when they must distinguish more categories.\footnote{E.g., \citet{stantonNosofsky2013category}.} Research on visual working memory also finds that people rely more heavily on categorical representations as the number of items they must remember increases. Related evidence suggests that categorical representations can reduce mental effort.\footnote{Read \citet{hardmanVergauweRicker2017categorical} for memory-load effects and \citet{zhouLoristMathot2022effort} for evidence from pupil dilation. \citet{baeEtAl2015colors} also find that recalled colors are biased toward category centers.}

How the available capacity is used matters in its own right. Experiments show that people can exploit statistical regularities to remember more items.\footnote{E.g., \citet{bradyKonkleAlvarez2009compression}. \citet{huttenlocherHedgesVevea2000categories,simsJacobsKnill2012ideal} also document how the distribution of stimuli affects recall. \citet{bowerEtAl1969hierarchical} find that organizing words into meaningful hierarchies improves recall relative to a scrambled presentation of the same words.} A capacity constraint forces information loss, but the organization of the representation determines which information is lost. And as I show, many forms of loss are easily avoidable by moving to a dominating summary.

\paragraph{The formal exercise.} I assume that there is a finite set of states and an exogenous source of information that induces a random posterior vector \(P\) with law \(\mu\). Given \(k \geq 1\), a law \(\nu\) is a \textit{summary} of \(\mu\) if it has at most \(k\) support points and \(\mu\) dominates it in convex order. Equivalently, a summary assigns the source posteriors--possibly at random--to at most \(k\) bins and replaces the posteriors assigned to each bin by their barycenter. A summary is \textit{maximal} if no other such summary yields a decision maker a weakly higher expected payoff in every decision problem and a strictly higher expected payoff in some. The \textit{improvement complexity} of \(\nu\) is the fewest support points in any posterior law that \(\mu\) dominates and that strictly dominates \(\nu\) in convex order.

My main result, \Cref{thm:maximal-summary-frontier}, says that, for regular \(\mu\),\footnote{The affine hull of the support has positive dimension, and \(\mu\) assigns zero probability to every proper affine hyperplane of that hull; see \Cref{def:regular-reference-law}.} a summary is maximal if and only if it is a hierarchical power diagram summary (HPD summary) with \(k\) support points, equivalently, if its improvement complexity is \(k+1\). Once a summary is maximal, a strict Blackwell improvement requires a larger vocabulary, and one additional message is always sufficient.

What is an HPD summary? Intuitively, think of a finite-action decision problem; its value function is a maximum of finitely many affine functions on the probability simplex. The full-dimensional regions where these functions attain the maximum form polytopes, or \textit{cells}, that constitute a \textit{power diagram}. A \textit{hierarchical power diagram} recursively subdivides selected cells using further power diagrams. Collapsing \(\mu\) on each final cell to its barycenter produces an \textit{HPD summary}.\footnote{\Cref{fig4} illustrates a power diagram. \Cref{fig2} is a hierarchical power diagram but not a power diagram. \Cref{fig1,fig3} are neither, though \Cref{fig1} is a convex partition.}

My characterization derives, then sharpens, a familiar restriction on categorization. I predict \textbf{convexity}, which requires that, whenever two beliefs belong to the same category, beliefs between them also belong to that category. Workhorse models of categorization impose this restriction.\footnote{\citet[Theorem~1]{mohlin2014optimal} restricts attention to convex categories; and \citet{gardenfors2004framework} postulates convex regions in ``conceptual space.''} Here, efficient use of a limited message vocabulary \textit{implies} convexity. But convexity is not all there is: maximal summaries are \textbf{hierarchical} classifications. They first separate beliefs into broad categories and then make fine distinctions within them. It is as if the decision-maker postulates a decision problem that is of first-order importance and arranges information at the highest level according to that--perhaps colors; red, blue, yellow. Within each of these highest-level categories, she then entertains decision problems of second order import--if it's red, she distinguishes between light and dark; if it's blue, between shiny, sparkly, or matte; and so on.\footnote{Or think of how Netflix organizes movies into genres and subgenres: comedies into romantic or black comedies, and action movies into spy or martial arts films. Amazon's book categories proceed from children's books to geography and cultures, then to ``Explore the World,'' and then to books about individual countries such as Mexico. eBay's categories proceed from collectibles to autographs, then political autographs, then presidential autographs. Google Shopping sorts electronics into computers and then tablet computers, while Shopify sorts home and garden products into decor, then clocks, then alarm clocks.} Within each of these, she moves to third-order considerations, etc. Hierarchical power diagrams have precisely this lexicographic interpretation.

\Cref{fig:amazon-hpd-stages} presents a stylized four-layer HPD summary corresponding to a particular sequence of refining categories and subcategories offered by the retailer Amazon: Children’s books \(\rightarrow\) Geography \& cultures \(\rightarrow\) Explore the World \(\rightarrow\) Mexico.

\begin{figure}[htbp]
    \centering
    \includegraphics[width=\linewidth]{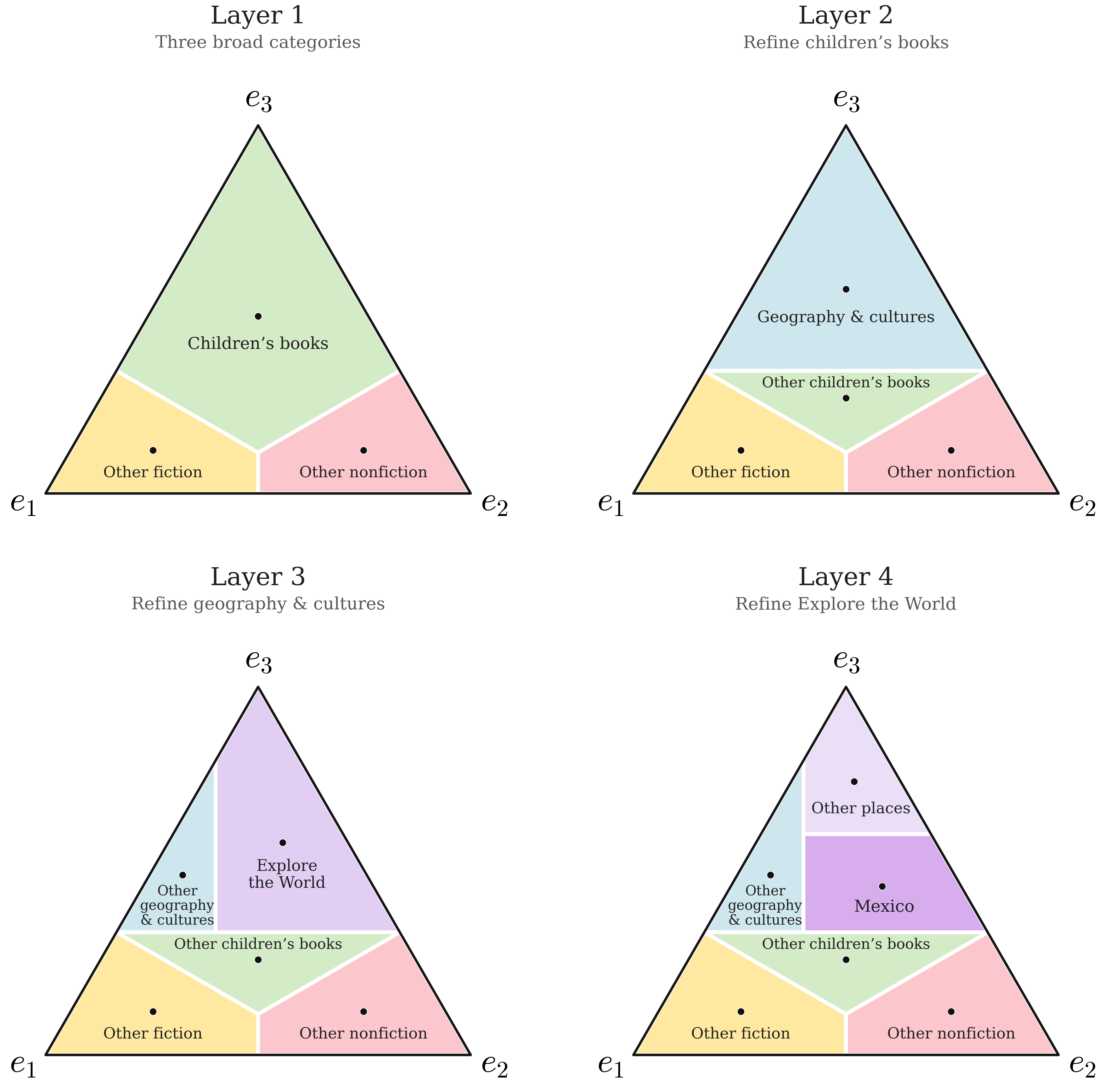}
    \caption{A four-layer HPD summary inspired by Amazon's book categories. The first layer divides the simplex into three cells: children's books, other fiction, and other nonfiction. The second layer subdivides children's books into Geography \& Cultures and other children's books. The third subdivides Geography \& Cultures into Explore the World and other geography and cultures. The fourth subdivides Explore the World into Mexico and other places. Each refinement retains the cells outside its parent. Black dots mark the cell barycenters under the uniform law.}
    \label{fig:amazon-hpd-stages}
\end{figure}

\paragraph{Approach.} To characterize the maximal summaries, I first study an auxiliary question. Fix not only the number of bins but also their probabilities. I consider all assignments of the original posteriors to bins with these probabilities. The resulting vectors of bin barycenters form a convex set--what are its extreme points? Precisely the vectors generated by hierarchical power diagrams (\Cref{prop:fixedweight}).

The connection to maximality begins with a nonextreme vector, for which I can produce the same summary by randomizing between two assignments with different barycenter vectors. Where the assignments disagree, some beliefs can be assigned to either of two bins. I can exchange small amounts of mass between these bins, keeping their probabilities fixed while moving their posterior means directly apart and thereby producing a strictly more informative summary. In short, I can improve information without adding a message or even changing how often any message is used (\Cref{lem:same-vocabulary-improvement}).

With more than one bin, an extreme vector maximizes a nonconstant linear criterion on this set. The resulting regions form a power diagram, the highest level of the ultimate hierarchical power diagram; in turn, the linear criterion is the highest-level decision problem in the lexicographic interpretation. At this level, a category may still contain several labels. In that case, this first comparison leaves their separate contents undetermined. But their vector of bin means must itself be extreme: otherwise, I could hold all other bins fixed and express the original vector as an average of two distinct vectors. And so I repeat the argument within that category--the second level of the hierarchy. Continuing until every final cell contains one message completes the hierarchy.

These successive payoff comparisons also render the assignment of beliefs to bins unique once their probabilities and means are fixed. Any alternative assignment with those probabilities and means must attain the same expected payoff at each stage and so respect every division. Now consider a more informative summary of \(\mu\). It can reproduce the original bin probabilities and means by randomly relabeling its messages, so it must reproduce that unique assignment. Relabeling cannot undo pooling across the original bins; and each new message must, accordingly, identify one original bin. With the same number of messages, nothing has changed. A strict improvement must split an existing bin and requires another message (\Cref{lem:hierarchical-rigidity}). Thus, once all \(k\) messages are in use, a hierarchical summary is maximal. One additional message always suffices for an improvement: reveal a further binary distinction within an existing bin (\Cref{lem:one-more-message}).

\subsection*{Related literature}
This is a particular exercise in \textit{coarse information design}. Earlier work in one-dimensional spaces shows that a convex value function admits an optimal finite-message policy that partitions the line into intervals.\footnote{See \citet{smithSorensenTian2021informational,tian2022optimal,lyuSuenZhang2026coarse}.} One dimension, however, hides the rich geometry I uncover: every finite convex partition is an interval partition, every interval partition is a power diagram, and recursive refinement produces another interval partition. \citet{aybasTurkel2026persuasion} allow a general finite state space but fix the persuasion objective, while \citet{kleinerMoldovanuStrackWhitmeyer2026} show that every finite summary can be weakly improved to a convex-partitional summary without increasing the number of messages.\footnote{For connections to choice under uncertainty and belief elicitation, see \citet{greenOsband1991revealed,lambert2019elicitation,frongilloKash2021truthfulness}. Power diagrams also arise in quadratic optimal transport from an absolutely continuous distribution to a finitely supported one \citep[Chapter~5]{galichon2016optimal}.}

Many papers study optimal categorization.\footnote{For convex categories, see \citet{mohlin2014optimal,azrieliLehrer2007categorization}; for Voronoi diagrams, see \citet{gardenfors2000conceptual,gardenfors2004framework,jager2007evolution,jagerMetzgerRiedel2011voronoi}.} \citet{jakobsen2026coarse} axiomatizes ``coarse Bayesian updating'' with convex belief categories; from this perspective, my work characterizes the Blackwell-maximal class of updating rules under support constraints.

Finally, in work independent of and concurrent with mine, \citet{bourne2026semidiscrete} study mean-preserving contractions with fixed cell masses and reference laws with bounded densities. They also characterize extreme barycenter vectors through hierarchical power diagrams (their Theorem~3.9 and my \Cref{prop:fixedweight}\eqref{it:fixedweight}) and establish uniqueness of the conditional laws given cell masses and barycenters (their Proposition~3.4 and my \Cref{lem:extreme-unique-disintegration}). The proofs are different; interestingly, almost reversed: they first establish uniqueness using optimal transport, then characterize extreme points. I instead prove the extreme-point characterization by induction and establish uniqueness by perturbation. They have no analog of my main result (\Cref{thm:maximal-summary-frontier}).

\section{Setup \& main result}\label{sec:regular-posterior-laws}

There is a finite state space \(\Theta\) and an exogenous source of information that generates a random posterior \(P\). Set \(\mu\coloneqq\operatorname{Law}(P)\), \(X\coloneqq\operatorname{co}(\operatorname{supp}\mu)\), and \(H\coloneqq\operatorname{aff}X\).\footnote{\(\operatorname{aff}X\) denotes the affine hull of \(X\), \textit{viz.,} the smallest affine space containing \(X\): \(\operatorname{aff}X\coloneqq\left\{\sum_{i=1}^n\alpha_ix_i\colon x_i\in X,\ \sum_{i=1}^n\alpha_i=1\right\}\). The coefficients may be negative.} Equip \(H\) with a Euclidean structure, and let \(\mathcal P(X)\) denote the set of Borel probability measures on \(X\). For \(\eta\in\mathcal P(X)\), write \(m_\eta\coloneqq\int_X p\,\eta(\mathrm dp)\).

\begin{assumption}\label{ass:regular-posterior}
The dimension \(d\coloneqq\dim H\) is positive, and \(\mu(L)=0\) for every proper affine hyperplane \(L\) of \(H\).\footnote{A \textit{proper affine hyperplane} \(L\) of \(H\) is a set \(L=\{p\in H\colon\langle a,p-p_0\rangle=0\}\) for some \(p_0\in H\) and nonzero \(a\in H-H\). The qualifier ``proper'' means that \(L\neq H\).}
\end{assumption}

I maintain \Cref{ass:regular-posterior} throughout.

\begin{definition}\label{def:regular-reference-law}
A law \(\rho\in\mathcal P(X)\) is a \textit{regular reference law} if \(\operatorname{aff}(\operatorname{supp}\rho)=H\) and \(\rho(L)=0\) for every proper affine hyperplane \(L\) of \(H\). In particular, \(\mu\) is a regular reference law under \Cref{ass:regular-posterior}.
\end{definition}

I recall the convex order: for \(\nu,\lambda\in\mathcal P(X)\), write \(\nu\preceq\lambda\) when \(\int_X\varphi(p)\nu(\mathrm dp)\leq\int_X\varphi(p)\lambda(\mathrm dp)\) for every continuous convex function \(\varphi\colon X\to\mathbb R\). Write \(\nu\prec\lambda\) when \(\nu\preceq\lambda\) and the inequality is strict for some such \(\varphi\).

\begin{definition}\label{def:summary-order}
A \textit{summary} of \(\mu\) is a law \(\nu\in\mathcal P(X)\) satisfying \(\nu\preceq\mu\). For \(k\geq1\), let
\[\mathcal S_k(\mu)\coloneqq \left\{\nu\in\mathcal P(X)\colon \nu\preceq\mu,\ \left|\operatorname{supp}\nu\right|\leq k\right\}.\]
A law \(\nu\in\mathcal S_k(\mu)\) is \textit{maximal} if there is no \(\lambda\in\mathcal S_k(\mu)\) with \(\nu\prec\lambda\). It is \textit{finite} if it has finite support.
\end{definition}

\begin{definition}\label{def:power-diagram-summary}
A \textit{power diagram} of \(H\) is a finite collection \(\mathcal K=(K_1,\ldots,K_r)\) of subsets of \(H\), \textit{cells}, constructed as follows. Choose affine functions \(\ell_1,\ldots,\ell_r\colon H\to\mathbb R\) with distinct linear parts, and let \(K_i\coloneqq \left\{x\in H\colon \ell_i(x)\geq\ell_j(x) \text{ for every }j\right\}\) (\(i=1,\ldots,r\)).
\end{definition}
Intuitively, each cell \(K_i\) contains the points at which the ``affine score'' \(\ell_i\) is maximal; and the cells cover \(H\).\footnote{Note that the cells may overlap where two or more scores tie. These overlaps are contained in finite unions of proper affine hyperplanes and are, consequently, null under every regular reference law.} If every \(K_i\) has nonempty interior in \(H\), then \(\mathcal K\) is an \(r\)-cell power diagram. Power diagrams are also called Laguerre diagrams or regular affine partitions.

\begin{definition}\label{def:hierarchical-power-diagram-summary}
A \textit{hierarchical power diagram} of \(H\) is a finite collection of cells obtained by beginning with a power diagram and then refining its cells a finite number of times. To refine a current cell \(K\), choose a power diagram \(\mathcal Q=(Q_1,\ldots,Q_s)\) of \(H\) and replace \(K\) by \(K\cap Q_1,\ldots,K\cap Q_s\), after discarding intersections with empty interior in \(H\).\end{definition}
Naturally, every power diagram is an unrefined hierarchical power diagram.

If \(\mathcal K=(K_1,\ldots,K_r)\) is a hierarchical power diagram and each cell has positive \(\mu\)-mass, define
\[z_i\coloneqq\mu(K_i), \qquad m_i\coloneqq\frac{1}{z_i}\int_{K_i}p\,\mu(\mathrm dp), \qquad \text{\&} \qquad \nu_{\mathcal K}\coloneqq\sum_{i=1}^r z_i\delta_{m_i}.\]
I call a finite summary a \textit{power diagram summary} if it equals \(\nu_{\mathcal K}\) for a power diagram \(\mathcal K\) whose cell barycenters are distinct. I call a finite summary a \textit{hierarchical power diagram summary (HPD summary)} if the same holds for a hierarchical power diagram.

\begin{definition}\label{def:improvement-complexity}
For a finite summary \(\nu\), define its \textit{improvement complexity} by
\[\kappa_\mu(\nu)\coloneqq \min \left\{\left|\operatorname{supp}\lambda\right|\colon \lambda\text{ is finite and } \nu\prec\lambda\preceq\mu\right\}.\]
The set in this definition is nonempty by \Cref{lem:one-more-message}.
\end{definition}

My main result is the following:
\begin{theorem}\label{thm:maximal-summary-frontier}
Let \(k\geq1\) and \(\nu\in\mathcal S_k(\mu)\). The following statements are equivalent.
\begin{enumerate}[noitemsep,nosep]
\item The summary \(\nu\) is maximal in \(\mathcal S_k(\mu)\).
\item The law \(\nu\) has exactly \(k\) support points and is an HPD summary.
\item The improvement complexity satisfies \(\kappa_\mu(\nu)=k+1\).
\end{enumerate}
\end{theorem}

\begin{figure}[p]
    \centering
    \captionsetup[subfigure]{font=small,skip=2pt}
    \begin{subfigure}[t]{0.47\textwidth}
        \centering
        \includegraphics[
            width=\linewidth,
            height=0.30\textheight,
            keepaspectratio
        ]{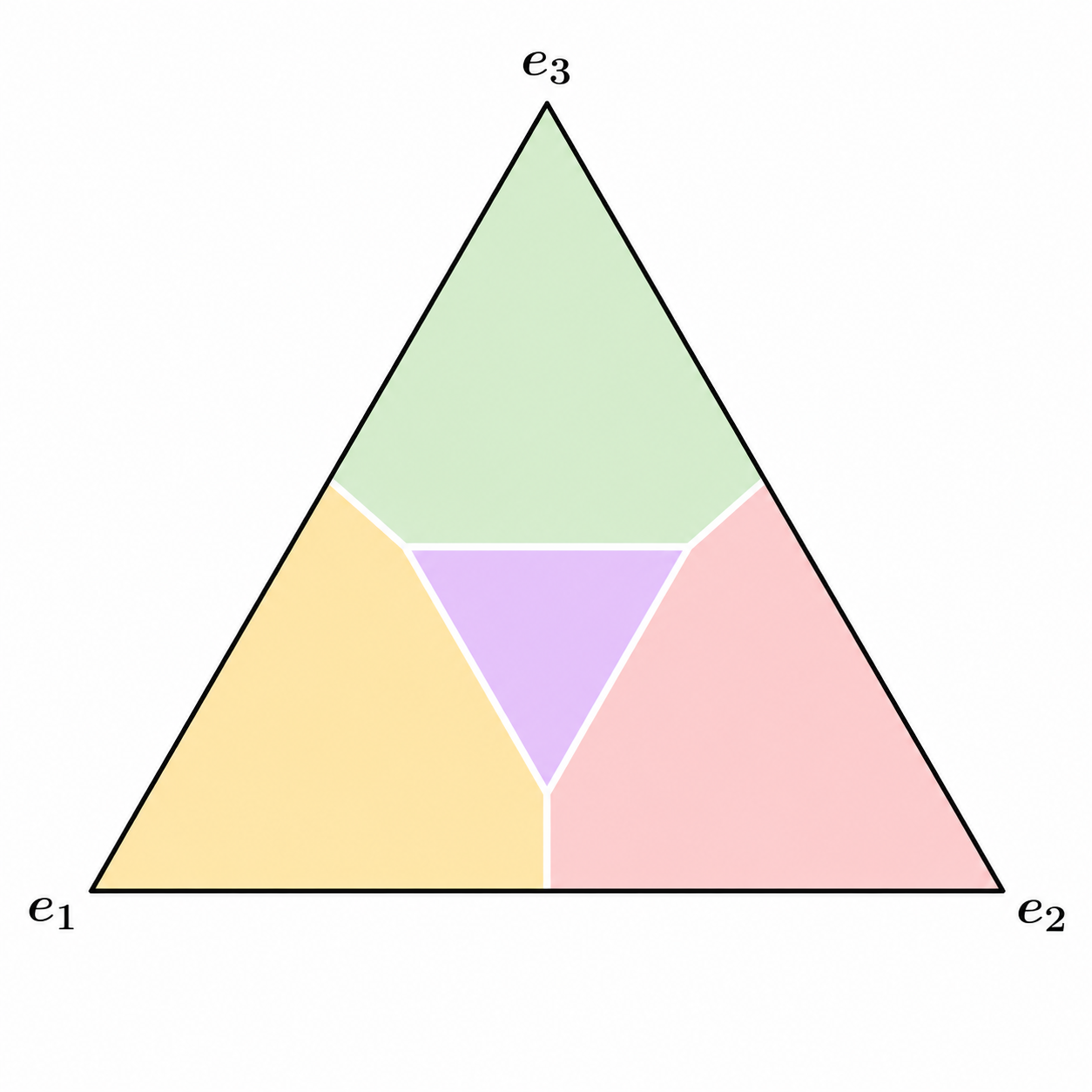}
        \caption{A Power Diagram}
        \label{fig4}
    \end{subfigure}
    \hfill
    \begin{subfigure}[t]{0.47\textwidth}
        \centering
        \includegraphics[
            width=\linewidth,
            height=0.30\textheight,
            keepaspectratio
        ]{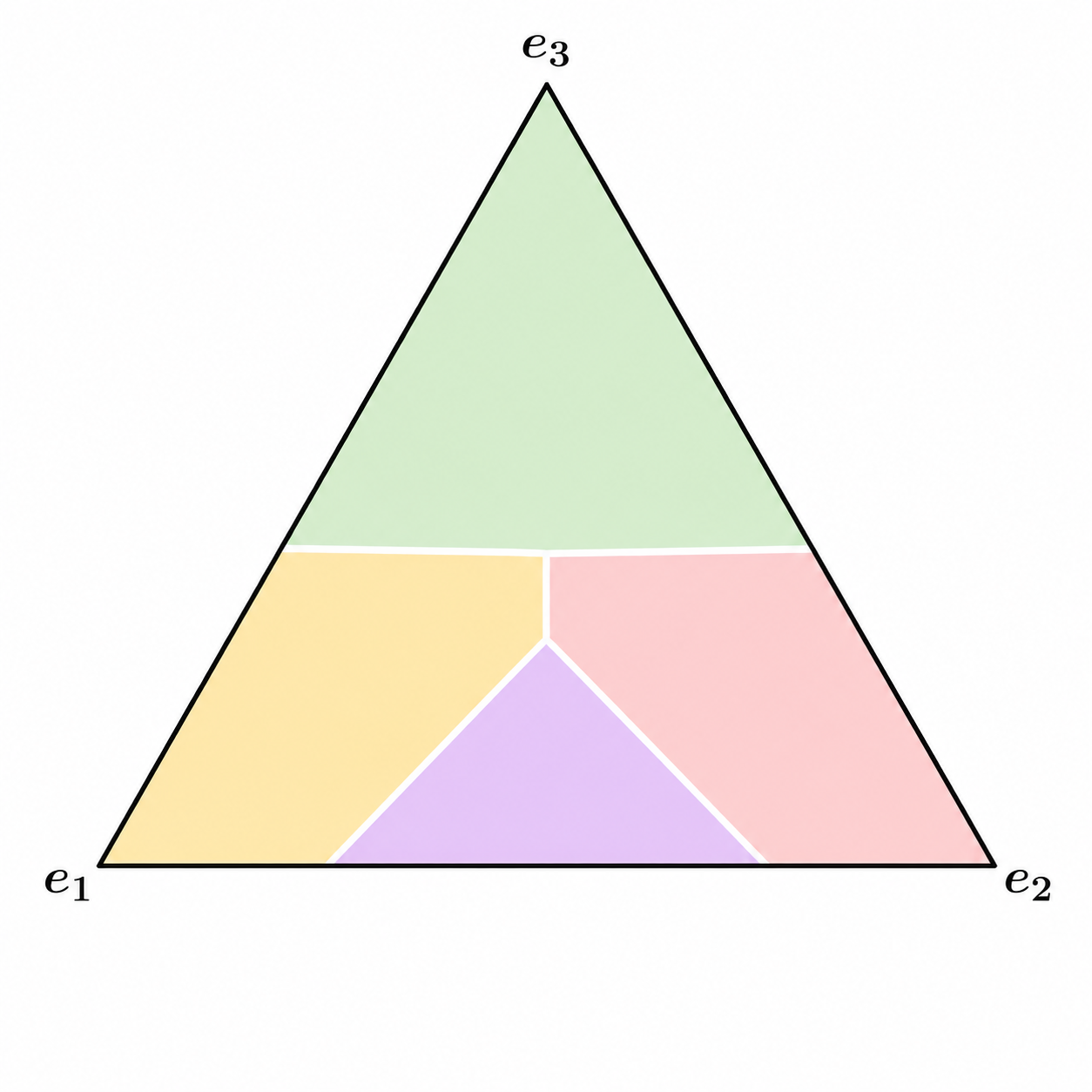}
        \caption{A Hierarchical Power Diagram}
        \label{fig2}
    \end{subfigure}

    \par\medskip

    \begin{subfigure}[t]{0.47\textwidth}
        \centering
        \includegraphics[
            width=\linewidth,
            height=0.30\textheight,
            keepaspectratio
        ]{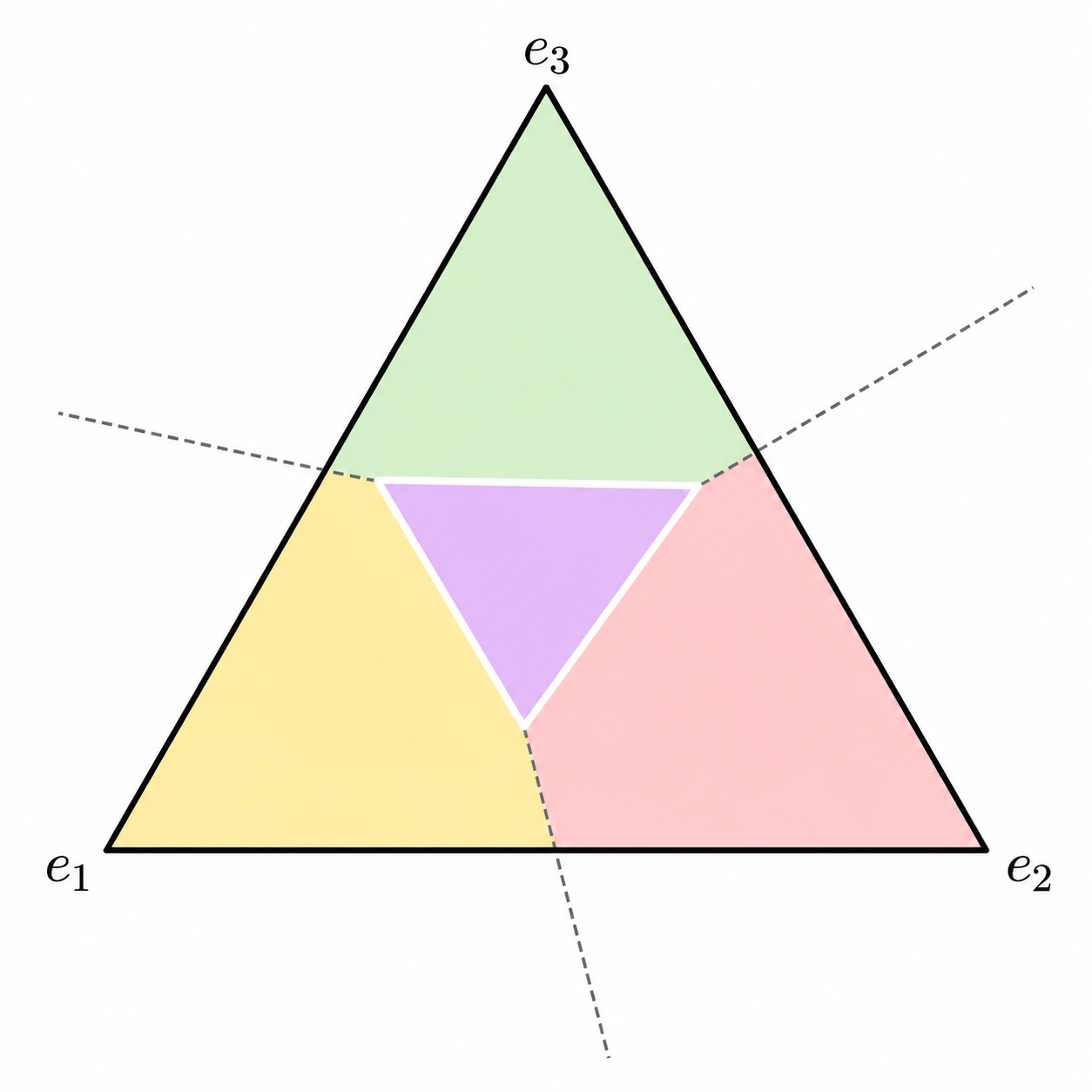}
        \caption{Not a Hierarchical Power Diagram}
        \label{fig1}
    \end{subfigure}
    \hfill
    \begin{subfigure}[t]{0.47\textwidth}
        \centering
        \includegraphics[
            width=\linewidth,
            height=0.30\textheight,
            keepaspectratio
        ]{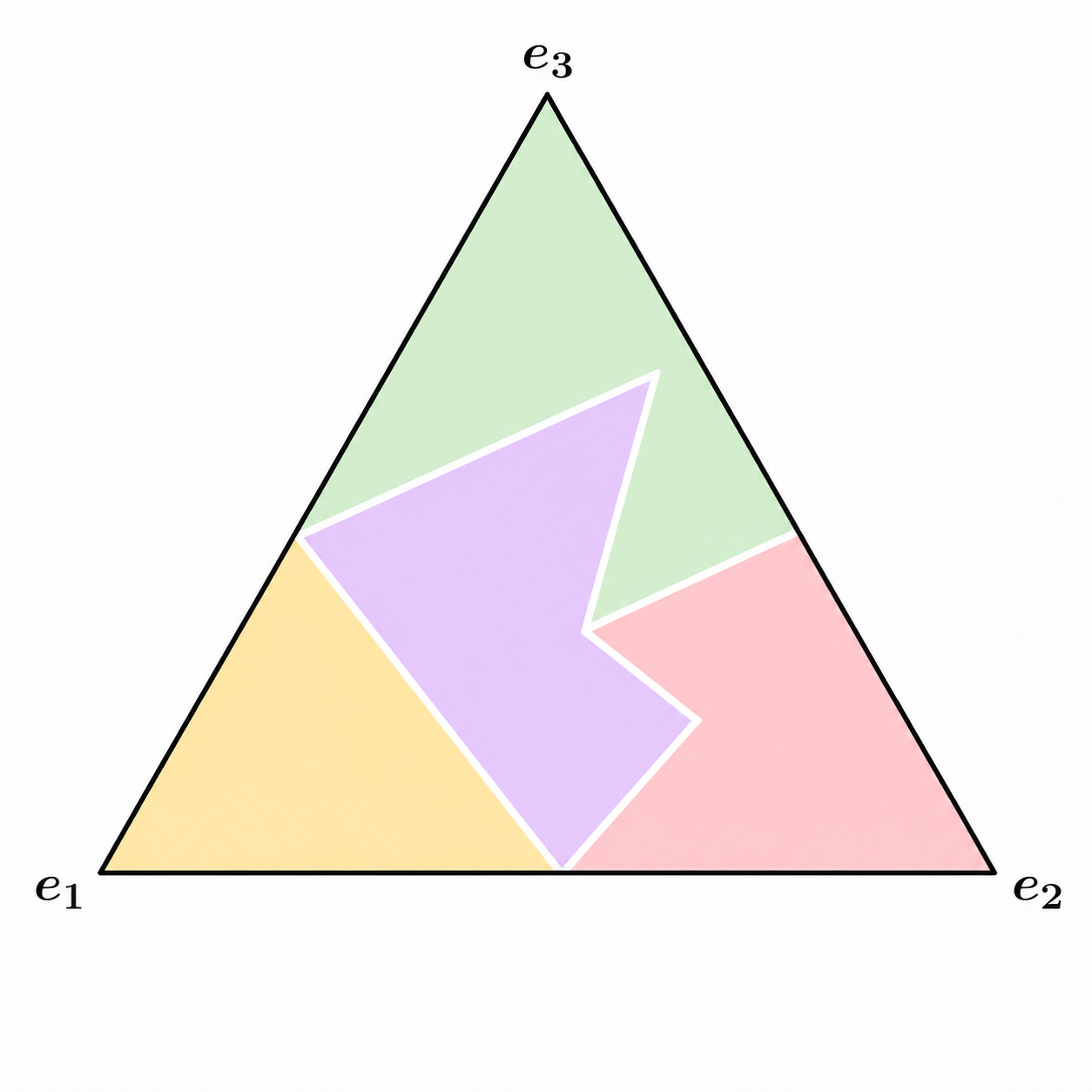}
        \caption{Not a Hierarchical Power Diagram}
        \label{fig3}
    \end{subfigure}
    \caption{Four ways of arranging information.}
    \label{fig:four-message-summary-examples}
\end{figure}

\section{An auxiliary problem}

My goal is to characterize, for a fixed \(k \geq 1\), the maximal summaries with support on \(k\) or fewer points of a regular law \(\mu\). To that end, in this section, I pursue an auxiliary question. Suppose we fix not only the number of messages \(r\) but also the probabilities of each message \(z_1,\ldots,z_r\) and consider the resulting vectors of bin barycenters for an arbitrary regular reference law \(\rho\).\footnote{I state this for some \(\rho\) not \(\mu\), because I apply the argument recursively to the conditional law within each cell of a power diagram.} This is a convex set--what are its extreme points? As we will see shortly, they are precisely the maximal summaries.

Formally, for \(r\geq1\), let \(\Delta_r^\circ\coloneqq \left\{z\in(0,\infty)^r\colon \sum_{i=1}^rz_i=1\right\}\). For \(z\in\Delta_r^\circ\) and \(M=(m_1,\ldots,m_r)\in X^r\), set \(\nu_z(M)\coloneqq\sum_{i=1}^rz_i\delta_{m_i}\). For any \(\rho\in\mathcal P(X)\), set \(\mathsf{MPC}_r(z,\rho)\coloneqq \left\{M\in X^r\colon \nu_z(M)\preceq\rho\right\}\).

I hold fixed both the number of messages and their probabilities. Accordingly, for \(M=(m_1,\ldots,m_r),N=(n_1,\ldots,n_r)\in\mathsf{MPC}_r(z,\rho)\) and \(t\in[0,1]\), their convex combination is \(tM+(1-t)N\coloneqq(tm_i+(1-t)n_i)_{i=1}^r\). I also note that Strassen's theorem \citep[Theorem 8]{strassen1965existence} implies that \(M=(m_1,\ldots,m_r)\) belongs to \(\mathsf{MPC}_r(z,\rho)\) if and only if there are probability measures \(\eta_1,\ldots,\eta_r\) such that\footnote{Here, \(m_{\eta_i}\coloneqq\int_Xp \eta_i(\mathrm dp)\) is the barycenter of \(\eta_i\).}
\[\sum_{i=1}^rz_i\eta_i=\rho \qquad \text{\&} \qquad m_{\eta_i}=m_i \quad\text{for every }i, \tag{1}\label{eq:disintegration}\]
which is, in turn, equivalent to the existence of measurable functions \(q_i\colon X\to[0,1]\) satisfying
\[\sum_{i=1}^rq_i=1\quad\rho\text{-a.e.}, \quad \int_Xq_i(p) \rho(\mathrm dp) =z_i, \quad \text{\&} \quad m_i =\frac{1}{z_i}\int_Xp q_i(p) \rho(\mathrm dp), \tag{2}\label{eq:density-rep}\]
since \(z_i\eta_i\leq\rho\). These two expressions are intuitive: the way to implement such a summary is to specify, for every original posterior \(p\) the probabilities \(q_i(p)\) of placing posterior \(p\) in bin \(i\). These probabilities must sum to one, they must produce the specified message probability \(z_i\), and I define \(m_i\) to be exactly the average posterior assigned to this bin.

The central result of this section reveals that the extreme points of \(\mathsf{MPC}_r(z,\rho)\) are exactly the hierarchical power diagram summaries.

\begin{proposition}\label{prop:fixedweight}
Suppose \(\rho\) is a regular reference law. For every \(r\geq1\) and \(z\in\Delta_r^\circ\), the following statements hold.
\begin{enumerate}[noitemsep,nosep]
\item\label{it:affine-hull} The set \(\mathsf{MPC}_r(z,\rho)\) is compact and convex, and
\[\operatorname{aff}\mathsf{MPC}_r(z,\rho)= \left\{M=(m_1,\ldots,m_r)\in H^r\colon \sum_{i=1}^rz_im_i=m_\rho\right\}.\]
\item\label{it:fixedweight} The extreme points of \(\mathsf{MPC}_r(z,\rho)\) are the barycenter vectors of hierarchical power diagrams with \(r\) cells and cell masses \(z\) under \(\rho\).
\end{enumerate}
\end{proposition}

\Cref{it:fixedweight} has a simple explanation. Fix the probabilities \(z_1,\ldots,z_r\) of the \(r\) messages and think of a vector \(M=(m_1,\ldots,m_r)\) as sorting the original posteriors among \(r\) bins: bin \(i\) must contain probability \(z_i\), and \(m_i\) is the average of the posteriors placed into that bin. The vector \(M\) is extreme if there are no two distinct vectors in \(\mathsf{MPC}_r(z,\rho)\) whose average is \(M\).

To see why an extreme vector must come from a hierarchical power diagram, begin with an extreme vector \(M\). Because \(M\) lies on the boundary of \(\mathsf{MPC}_r(z,\rho)\), there is a nonconstant linear criterion that \(M\) maximizes, perhaps together with other vectors in \(\mathsf{MPC}_r(z,\rho)\). For each posterior \(p\), this criterion assigns a ``score'' to placing \(p\) into each bin--of course, it may assign exactly the same score to several bins at every \(p\). In any case, group together the bins that it treats identically. Importantly, the criterion can distinguish between the different groups, but it says nothing about how the posteriors sent to one group should be divided among the bins in that group.

Each group must receive a fixed total probability. To see how this restriction determines the first division of the posteriors, suppose first that there are only two groups, \(A\) and \(B\). For each posterior \(p\), compare the score from placing \(p\) in \(A\) with the score from placing it in \(B\). Among all sortings producing elements of \(\mathsf{MPC}_r(z,\rho)\), a maximizer sends to \(A\) the prescribed share of posteriors for which \(A\)'s advantage over \(B\) is greatest. Otherwise, some posterior mass sent to \(A\) would have a smaller advantage than an equal amount sent to \(B\), and exchanging the two would preserve the sizes of the groups and raise the criterion.

The division between \(A\) and \(B\) is, accordingly, determined by a cutoff in \(A\)'s advantage over \(B\). Since that advantage is affine in \(p\), the boundary is an affine hyperplane. With more than two groups, I show in the proof that the corresponding cutoffs can be chosen simultaneously so that every group receives its required total probability. Each posterior is assigned to the group whose affine score wins these comparisons, producing precisely a power diagram.

But this first power diagram need not finish the sorting. If a group contains several bins, the criterion says only that the posteriors in its cell must be assigned somewhere among those bins and, crucially, not how they should be divided among them. Inside that cell, there remains a smaller problem of exactly the same kind: sort its posteriors among fewer bins, with the probability of each bin fixed.

As \(M\) is extreme, it must be that the vectors of bin averages inside every first-stage cell are extreme in this smaller problem. Consequently, I can repeat the argument inside every cell that still contains more than one bin. A first power diagram makes the broad divisions; further power diagrams divide the posteriors within those regions; and later power diagrams may divide them again within the cells created earlier. Each step leaves fewer bins to be separated, so the procedure ends after finitely many steps, yielding at last a hierarchical power diagram.

Conversely, suppose that \(M\) is induced by a hierarchical power diagram and is the average of two distinct vectors in \(\mathsf{MPC}_r(z,\rho)\). The linear criterion associated with the first-stage power diagram is maximized at \(M\). Since the criterion is linear, both vectors in the average must maximize it as well and must respect the same first-stage division into groups. Within each first-stage cell, the remaining refinements form a smaller hierarchical power diagram. By induction, its vector of bin averages is extreme in the corresponding smaller problem, so the two vectors must agree within every first-stage cell. They must, therefore, both equal \(M\).

\section{Proof of the main theorem \texorpdfstring{(\Cref{thm:maximal-summary-frontier})}{theorem}}

I now use \Cref{prop:fixedweight} to characterize the maximal MPCs of \(\mu\) with a fixed number of messages in support. I now show that every nonextreme \(M\in\mathsf{MPC}_r(z,\rho)\) (whose coordinates are distinct) has some \(M'\in\mathsf{MPC}_r(z,\rho)\) with \(\nu_z(M)\prec\nu_z(M')\).

I establish this in two steps. First, I adapt an argument from \citet{kleinerMoldovanuStrackWhitmeyer2026} to show that, whenever a representation sometimes sends the same original posterior to each of two messages on a set of positive probability, I can shift a small amount of probability between those messages in opposite directions. This leaves every message probability unchanged but moves the two reported posteriors farther apart (which is a strict improvement in informativeness). Second, I show that every nonextreme \(M\) has such a representation. Indeed, if \(M\) is the midpoint of two distinct vectors, averaging one representation of each produces the required overlap wherever the two representations differ.

The idea behind the next lemma is simple. Suppose a representation randomizes between messages \(i\) and \(j\) on a positive-probability set of original posteriors. On this set, there is room to transfer a small amount of probability from \(j\) to \(i\) at some posteriors and from \(i\) to \(j\) at others. I choose the two transfers to have the same total probability, so that the probabilities of the two messages are unchanged, but their barycenters move (because the posteriors transferred in the two directions have different averages). Regularity ensures that the posteriors in the overlapping set contain enough variation to make the barycenters move in any prescribed direction.

\begin{lemma}\label{lem:fixed-probability-transfer}
Suppose \(\rho\) is a regular reference law, \(z\in\Delta_r^\circ\), and \(M=(m_1,\ldots,m_r)\in\mathsf{MPC}_r(z,\rho)\). Let \((q_h)_{h=1}^r\) satisfy \eqref{eq:density-rep}, and suppose that, for some distinct \(i\) and \(j\), the set \(\{p\in X\colon q_i(p)>0,\ q_j(p)>0\}\) has positive \(\rho\)-probability. For every \(\xi\in H-H\) and every sufficiently small \(\varepsilon>0\), there are vectors \(M^+,M^-\in\mathsf{MPC}_r(z,\rho)\) satisfying \(M=(M^++M^-)/2\) and
\[m_h^\pm=m_h\quad\text{for }h\notin\{i,j\}, \qquad m_i^\pm=m_i\pm\frac{\varepsilon}{z_i}\xi, \qquad \text{\&} \qquad m_j^\pm=m_j\mp\frac{\varepsilon}{z_j}\xi.\]
If \(m_i\neq m_j\), choosing \(\xi=m_i-m_j\) yields \(\nu_z(M)\prec\nu_z(M^+)\).
\end{lemma}

\Cref{lem:fixed-probability-transfer} has another immediate consequence. If a representation of \(M\) assigned a positive-probability set of original posteriors to two messages with positive probability, the lemma would express \(M\) as the average of two distinct vectors. Consequently, every representation of an extreme vector assigns almost every original posterior to a single message. Moreover, such a representation must be unique: averaging two different representations would produce one that assigns some posteriors to more than one message. I now state this formally:

\begin{lemma}\label{lem:extreme-unique-disintegration}
Suppose \(\rho\) is a regular reference law, \(z\in\Delta_r^\circ\), and \(M=(m_1,\ldots,m_r)\) is an extreme point of \(\mathsf{MPC}_r(z,\rho)\). There is a unique family of probability measures \(\eta_1,\ldots,\eta_r\) satisfying \eqref{eq:disintegration}. Moreover, there is a measurable partition \((D_1,\ldots,D_r)\) of \(X\), up to \(\rho\)-null sets, such that \(\eta_i=\rho|_{D_i}/z_i\) for every \(i\). If \(M\) is induced by a hierarchical power diagram \((K_1,\ldots,K_r)\), then \(D_i=K_i\) \(\rho\)-almost everywhere for every \(i\).
\end{lemma}

My next result is the crucial bridge connecting \Cref{prop:fixedweight} to maximal summaries. Namely, nonextremality leaves room to improve the summary without adding a message. To wit, if \(M\) is the average of two different vectors, take one assignment of the original posteriors producing each vector and average the two assignments. Where the assignments differ, the average sometimes sends the same posterior to two different messages. \Cref{lem:fixed-probability-transfer} then sorts those posteriors more sharply between the two messages, preserving the probability of every message but moving their distinct barycenters farther apart, yielding a strictly more informative summary that still uses the same number of messages. Formally:
\begin{lemma}\label{lem:same-vocabulary-improvement}
Suppose \(\rho\) is a regular reference law. Fix \(z\in\Delta_r^\circ\), and suppose \(M\in\mathsf{MPC}_r(z,\rho)\) has distinct coordinates and is not an extreme point. Every neighborhood of \(M\) contains a vector \(M'\in\mathsf{MPC}_r(z,\rho)\) with distinct coordinates such that \(\nu_z(M)\prec\nu_z(M')\), i.e., there is a strict improvement that uses the same message probabilities and the same number of messages.
\end{lemma}

An HPD summary is extreme by \Cref{it:fixedweight} of \Cref{prop:fixedweight}, so \Cref{lem:extreme-unique-disintegration} makes its assignment of original posteriors to messages unique. I now show that every strict improvement requires another message and then show that one additional message always suffices.

\begin{lemma}\label{lem:hierarchical-rigidity}
Let \(\nu=\sum_{i=1}^rz_i\delta_{m_i}\) be induced by an \(r\)-cell hierarchical power diagram \((K_1,\ldots,K_r)\), and suppose the \(m_i\) are distinct. If a finite summary \(\lambda\) satisfies \(\nu\preceq\lambda\preceq\mu\) and \(\lambda\neq\nu\), then \(\left|\operatorname{supp}\lambda\right|\geq r+1\).
\end{lemma}

\begin{lemma}\label{lem:one-more-message}
Every \(r\)-point summary \(\nu\) admits a finite summary \(\lambda\) such that \(\nu\prec\lambda\preceq\mu\) and \(\left|\operatorname{supp}\lambda\right|\leq r+1\). If \(\nu\) is an HPD summary, then \(\lambda\) may be chosen as an HPD summary and has exactly \(r+1\) support points.
\end{lemma}

Together, \Cref{lem:same-vocabulary-improvement,lem:hierarchical-rigidity,lem:one-more-message} yield a simple test for maximality. A summary using fewer than \(k\) posterior values can be improved while remaining within the limit. A \(k\)-point summary that is not generated by a hierarchical power diagram can be improved without adding a posterior value. A \(k\)-point HPD summary cannot be strictly improved unless a \((k+1)\)st posterior value is allowed.

I now complete the proof of the main theorem:
\begin{proof}[Proof of \Cref{thm:maximal-summary-frontier}]
Let \(r\coloneqq\left|\operatorname{supp}\nu\right|\), write \(\nu=\sum_{i=1}^rz_i\delta_{m_i}\) with distinct \(m_i\), and set \(z\coloneqq(z_1,\ldots,z_r)\) and \(M\coloneqq(m_1,\ldots,m_r)\).

Suppose first that \(\nu\) is maximal in \(\mathcal S_k(\mu)\). If \(r<k\), \Cref{lem:one-more-message} provides a strict improvement with at most \(r+1\leq k\) support points, contradicting maximality. Hence, \(r=k\). If \(M\) were not extreme in \(\mathsf{MPC}_k(z,\mu)\), \Cref{lem:same-vocabulary-improvement} would provide a strict improvement with at most \(k\) support points, again contradicting maximality. Therefore, \(M\) is extreme. By \Cref{it:fixedweight} of \Cref{prop:fixedweight}, \(\nu\) is an HPD summary.

Suppose next that \(\nu\) has exactly \(k\) support points and is an HPD summary. By \Cref{lem:hierarchical-rigidity}, every strict finite improvement has at least \(k+1\) support points. Thus, \(\nu\) is maximal in \(\mathcal S_k(\mu)\). Combining the lower bound from \Cref{lem:hierarchical-rigidity} with \Cref{lem:one-more-message} produces \(\kappa_\mu(\nu)=k+1\).

Finally, suppose \(\kappa_\mu(\nu)=k+1\). By \Cref{lem:one-more-message}, \(k+1=\kappa_\mu(\nu)\leq r+1\). Since \(r\leq k\), \(r=k\). If \(M\) were not extreme in \(\mathsf{MPC}_k(z,\mu)\), \Cref{lem:same-vocabulary-improvement} would imply \(\kappa_\mu(\nu)\leq k\), a contradiction. Hence, \(M\) is extreme, and \Cref{it:fixedweight} of \Cref{prop:fixedweight} implies that \(\nu\) is an HPD summary.
\end{proof}

\subsection{Understanding hierarchical power diagrams}

We can understand hierarchical power diagrams through their characterization as extreme points for the fixed-weight barycenter vectors. For \(\mu\) uniform on the simplex, \Cref{fig:nonextreme-decompositions} illustrates why the summaries in \Cref{fig1,fig3} are not HPDs, \textit{viz.}, have nonextreme barycenter vectors. In each row, the first two diagrams generate distinct extreme vectors with the same bin probabilities as the corresponding summary. The midpoint of these two vectors is the barycenter vector of that summary, as the black dots in the third panel show. This also implies that the assignment of posteriors to bins in the non-HPD summaries is not unique. Simply toss a fair coin to choose between the first two partitions in a row and report the resulting label, without revealing the coin. Both partitions use each label with the same probability, so the posterior following each label is the average of its two displayed barycenters, producing exactly the same summary as the original partition. But wherever the two partitions disagree, this rule randomizes between labels, producing a different implementation of that summary.

\begin{figure}[htbp]
    \centering
    \captionsetup[subfigure]{font=small,skip=2pt}
    \begin{subfigure}{\textwidth}
        \centering
        \includegraphics[width=\linewidth,trim=0 90bp 0 0,clip]{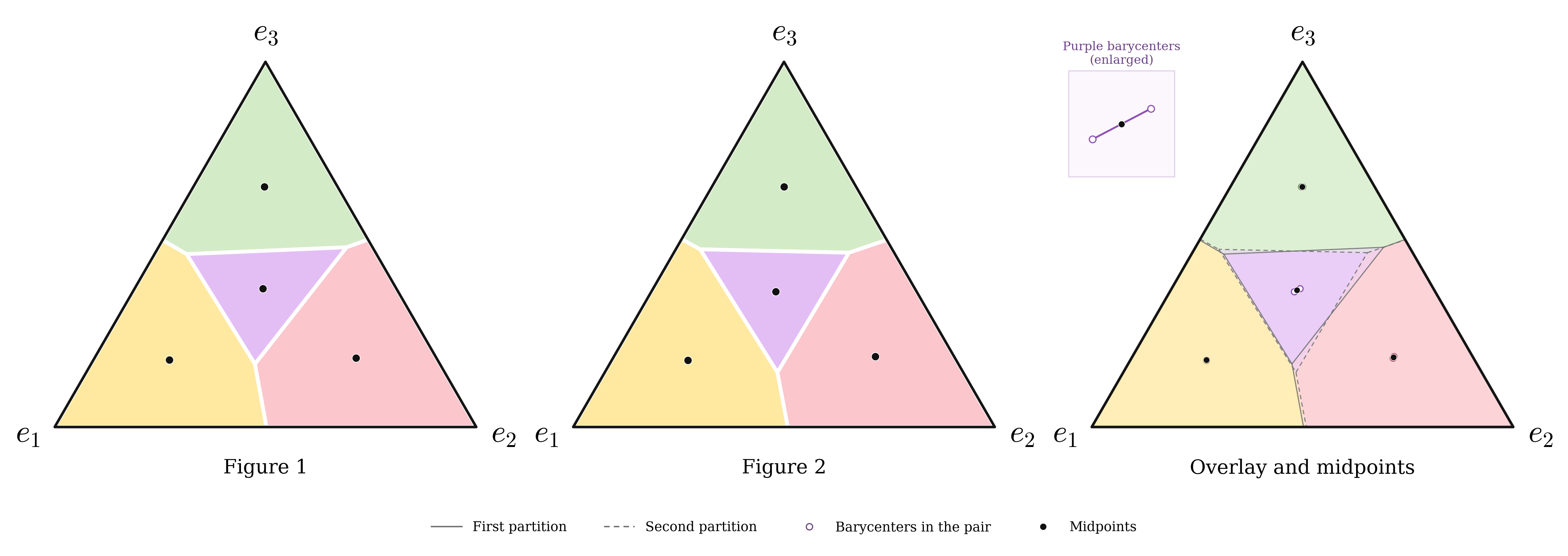}
        \caption{The summary in \Cref{fig1}.}
        \label{fig:convex-nonextreme-decomposition}
    \end{subfigure}

    \par\medskip

    \begin{subfigure}{\textwidth}
        \centering
        \includegraphics[width=\linewidth,trim=0 90bp 0 0,clip]{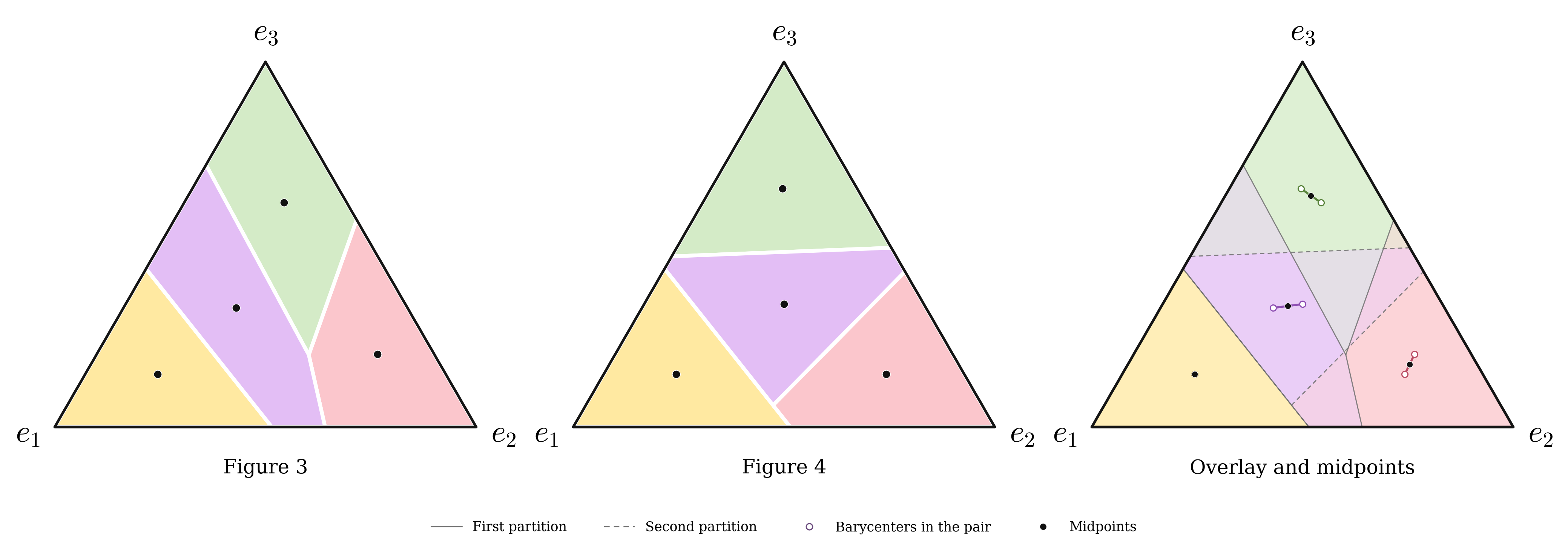}
        \caption{The summary in \Cref{fig3}.}
        \label{fig:nonconvex-nonextreme-decomposition}
    \end{subfigure}
    \caption{Nonextreme barycenter vectors, with \(\mu\) uniform on the simplex. Dots in the first two panels mark the bin barycenters. The third panel overlays the partitions: solid lines show the first and dashed lines the second. Hollow dots mark their barycenters, and black dots mark the corresponding midpoints.}
    \label{fig:nonextreme-decompositions}
\end{figure}

I can also use this randomization to improve the summary. Take a positive-area region inside the simplex where the first partition assigns label \(i\) and the second assigns label \(j\neq i\), and let \(m_i,m_j\) be the corresponding barycenters of the original summary. Inside this region, I choose two sufficiently small, disjoint disks of equal radius, with centers \(x_+\) and \(x_-\) satisfying \(x_+-x_-=c(m_i-m_j)\) for some \(c>0\). I send every posterior in the first disk to \(i\) and every posterior in the second to \(j\). Since \(\mu\) is uniform, the disks have the same probability and their barycenters are their centers. This change preserves the probability of each bin while moving \(m_i\) and \(m_j\) directly apart along the line joining them. Their weighted average and the other barycenters remain unchanged, so the new summary is strictly more informative while using the same number of messages.

\citet[Proposition~2]{kleinerMoldovanuStrackWhitmeyer2026} establish that maximal summaries must be convex partitional. The first row illustrates why this necessary condition is not sufficient: even a convex partition permits a strict improvement without changing the probability of any message. My characterization identifies the additional restriction--maximal summaries must have the hierarchical structure.

\section{Optimal summaries for a fixed decision problem}

The dominated summaries are precisely those that can be improved upon in every decision problem, strictly in some. I could instead fix a decision problem, and ask what summaries are optimal \textit{for that particular decision problem}. This criterion now singles out the \textit{power diagram summaries} as follows.

For a finite set of actions \(A\) and payoff function \(u\colon A\times\Theta\to\mathbb R\), define the value function (in the posterior \(p\)) \(W_u(p)\coloneqq \max_{a\in A} \sum_{\theta\in\Theta}p(\theta)u(a,\theta)\), and the value of a posterior law \(\eta\) by \(\mathcal V_u(\eta)\coloneqq\int_XW_u(p) \eta(\mathrm dp)\). Given \(k \geq 1\), I assume that there are at least \(k\) actions in \(A\) that are strictly optimal for some posterior \(p \in \operatorname{supp}\mu\). Then,
\begin{corollary}\label{cor:optimal-k-message-summary}
Fix \(k\geq1\) and a finite-action decision problem \(u\). If a summary \(\nu\in\mathcal S_k(\mu)\) maximizes \(\mathcal V_u\) over \(\mathcal S_k(\mu)\), it is a power diagram summary with exactly \(k\) points in support.
\end{corollary}
\begin{proof}
Take an optimal summary. It must use all \(k\) messages, since otherwise another message could identify a strictly-positive-probability set of original posteriors at which an unused action is strictly optimal. Choose an optimal action following each message. These \(k\) actions must be distinct, since otherwise I could merge messages without changing the payoff. Take the resulting \(k\)-action decision problem, whose regions of optimality form a power diagram in \(X\). The original summary is optimal in this reduced problem, since any improvement would also improve its payoff in the original problem. Regularity makes the optimal action unique \(\mu\)-almost everywhere, so optimality requires assigning each original posterior to the message corresponding to its optimal action. Accordingly, the original summary must be the corresponding power diagram summary.
\end{proof}

\bibliography{sample}

\appendix

\section[appendix]{Omitted Proofs}

\subsection{Proof of \texorpdfstring{\Cref{prop:fixedweight}}{prop}}
\begin{proof}[Proof of \Cref{prop:fixedweight}]
Let \(E\coloneqq H-m_\rho\) be the translation space of \(H\), and use the Euclidean structure on \(H\) to identify \(E\) with its dual.\footnote{The translation space \(E=H-m_\rho\) is the linear space of directions parallel to the affine space \(H\). Write \(\langle\cdot,\cdot\rangle\) for its Euclidean inner product, which associates each \(a\in E\) with the linear functional \(x\mapsto\langle a,x\rangle\); in finite dimensions, every linear functional on \(E\) arises uniquely in this way.}

Suppose that \(M^n=(m_1^n,\ldots,m_r^n)\in\mathsf{MPC}_r(z,\rho)\) and \(M^n\to M=(m_1,\ldots,m_r)\) in \(X^r\). For every continuous convex function \(\varphi\colon X\to\mathbb R\),
\[\sum_{i=1}^rz_i\varphi(m_i) = \lim_{n\to\infty}\sum_{i=1}^rz_i\varphi(m_i^n) \leq \int_X\varphi(p)\rho(\mathrm dp).\]
Hence, \(M\in\mathsf{MPC}_r(z,\rho)\), so \(\mathsf{MPC}_r(z,\rho)\) is closed in the compact set \(X^r\), \textit{ergo} compact. Convexity follows from \eqref{eq:density-rep}.\footnote{If \(M\) and \(N\) are implemented by \(q\) and \(\widetilde q\), respectively, then \(tq+(1-t)\widetilde q\) implements \(tM+(1-t)N\) for every \(t\in[0,1]\).}

\(\nu_z(M)\preceq\rho\) implies \(\sum_{i=1}^rz_i\ell(m_i) = \int_X\ell(p)\,\rho(\mathrm dp)\) for every affine \(\ell\). Since affine functionals separate points of \(H\), this equality implies
\[\sum_{i=1}^rz_im_i=m_\rho, \tag{3}\label{eq:barycenter-constraint}\]
and so \(\mathsf{MPC}_r(z,\rho)\) is contained in the affine space in \Cref{it:affine-hull}.

For the reverse inclusion, define \(\Gamma\colon E\to E\) by \(\Gamma\xi\coloneqq\int_X\left\langle\xi,p-m_\rho\right\rangle\left(p-m_\rho\right) \rho(\mathrm dp)\). Since \(\langle\xi,\Gamma\xi\rangle=\int_X\langle\xi,p-m_\rho\rangle^2\,\rho(\mathrm dp)\), \(\langle\xi,\Gamma\xi\rangle=0\) implies that the continuous function \(p\mapsto\langle\xi,p-m_\rho\rangle\) vanishes \(\rho\)-almost everywhere and, therefore, on \(\operatorname{supp}\rho\). \(\operatorname{aff}(\operatorname{supp}\rho)=H\) then implies \(\xi=0\), so \(\Gamma\) is invertible.

Let \(D=(d_1,\ldots,d_r)\in E^r\) satisfy \(\sum_{i=1}^rz_id_i=0\). Since \(X\) is compact and \(\Gamma^{-1}\) is continuous, there is \(\delta>0\) such that
\[\max_{1\leq i\leq r}\lVert d_i\rVert<\delta \qquad \Longrightarrow \qquad \left|\left\langle\Gamma^{-1}d_i,p-m_\rho\right\rangle\right|\leq1 \quad \forall i, \ \forall p\in X.\]
For every such \(D\), the functions
\[q_i^D(p)\coloneqq z_i\left(1+\left\langle\Gamma^{-1}d_i,p-m_\rho\right\rangle\right)\]
are nonnegative and satisfy
\[\sum_{i=1}^rq_i^D=1, \qquad \int_Xq_i^D(p) \rho(\mathrm dp)=z_i, \qquad \text{\&} \qquad \frac{1}{z_i}\int_Xp q_i^D(p) \rho(\mathrm dp)=m_\rho+d_i.\]
Hence, \(\mathsf{MPC}_r(z,\rho)\) contains a relative neighborhood of \((m_\rho,\ldots,m_\rho)\) in the affine space defined by \eqref{eq:barycenter-constraint}. We conclude \Cref{it:affine-hull}.

Let \(M^0\coloneqq(m_\rho,\ldots,m_\rho)\). Translating the affine hull by \(M^0\) identifies it with the vector space
\[V\coloneqq \left\{D=(d_1,\ldots,d_r)\in E^r\colon \sum_{i=1}^rz_id_i=0\right\}.\]
Every linear functional on \(V\) extends to a linear functional on \(E^r\). Since every \(z_i\) is positive, its coefficients may be written as \(z_ia_i\) for vectors \(a_i\in E\). Thus, every linear functional on the affine hull, after taking \(M^0\) as zero, can be written as
\[\Lambda_a(M)\coloneqq\sum_{i=1}^rz_i\left\langle a_i,m_i-m_\rho\right\rangle, \qquad \text{for } a_i\in E. \tag{4}\label{eq:exposing-functional}\]
The face exposed by \(\Lambda_a\) is the set of vectors in \(\mathsf{MPC}_r(z,\rho)\) that maximize \(\Lambda_a\).

Fix \(a_1,\ldots,a_r\in E\). Let \(I_1,\ldots,I_s\) partition \(\{1,\ldots,r\}\) according to the distinct values among these vectors, let \(a^h\) be the common value on \(I_h\), and set
\[w_h\coloneqq\sum_{i\in I_h}z_i \qquad \text{\&} \qquad z^h\coloneqq\left(\frac{z_i}{w_h}\right)_{i\in I_h}.\]
Define \(u_h(p)\coloneqq\langle a^h,p-m_\rho\rangle\), and write \(F_a\) for this face. The next claim is standard in semidiscrete optimal transport \citep[Chapter~5]{galichon2016optimal}. I include a proof under the regularity condition used here.

\begin{claim}\label{clm:prescribed-mass-diagram}
There is \(c^*=(c_1^*,\ldots,c_s^*)\in\mathbb R^s\) such that the affine scores \(p\mapsto u_h(p)+c_h^*\), for \(h=1,\ldots,s\), generate an \(s\)-cell power diagram \((G_1,\ldots,G_s)\) satisfying \(\rho(G_h)=w_h\) for every \(h\).
\end{claim}
\begin{proof}[Proof of \Cref{clm:prescribed-mass-diagram}]
Define
\[\Psi(c)\coloneqq \int_X\max_h\left\{u_h(p)+c_h\right\}\,\rho(\mathrm dp) -\sum_{h=1}^sw_hc_h, \qquad\text{for }c\in\mathbb R^s.\]
If \(s=1\), take \(c_1^*=0\) and \(G_1=H\). Suppose henceforth that \(s\geq2\). The function \(\Psi\) is continuous, convex, and invariant under common translations of \(c\).\footnote{For every \(t\in\mathbb R\), adding \(t\) to every coordinate of \(c\) adds \(t\) to the pointwise maximum and subtracts \(t\sum_{h=1}^sw_h=t\) in the second term. Hence, \(\Psi(c+t\mathbf 1)=\Psi(c)\), where \(\mathbf 1=(1,\ldots,1)\).} Let \(W\coloneqq\{c\in\mathbb R^s\colon\max_hc_h=0\}\). For every \(c\in W\), choose \(j\) with \(c_j=0\). Since \(\int_Xu_j(p)\,\rho(\mathrm dp)=0\), \(\Psi(c)\geq-\sum_{h=1}^sw_hc_h\).

Every coordinate of \(c\) is nonpositive. Consequently, \(\lVert c\rVert\to\infty\) within \(W\) implies \(-\sum_{h=1}^sw_hc_h\to\infty\), and so \(\Psi(c)\to\infty\). Since \(\Psi\) is continuous and coercive on the closed set \(W\), it has a minimizer \(c^*\) on \(W\). Every vector in \(\mathbb R^s\) differs by a common translation from a vector in \(W\), so translation invariance makes \(c^*\) a global minimizer.

Because the \(a^h\) are distinct, every set of ties
\[\left\{p\in H\colon u_h(p)+c_h^*=u_g(p)+c_g^*\right\}, \qquad\text{with }h\neq g,\]
is a proper affine hyperplane and, therefore, \(\rho\)-null. For \(h=1,\ldots,s\), set
\[G_h\coloneqq \left\{p\in H\colon u_h(p)+c_h^*\geq u_g(p)+c_g^* \text{ for every }g\right\}.\]
Fix \(h\). For \(t\neq0\), let
\[D_{h,t}(p)\coloneqq \frac{\max_g\{u_g(p)+c_g^*+t\1_{\{g=h\}}\}-\max_g\{u_g(p)+c_g^*\}}{t}.\]
Only the \(h\)th term inside the maximum changes and by \(t\). The absolute value of the numerator, therefore, is at most \(\lvert t\rvert\), so \(\lvert D_{h,t}(p)\rvert\leq1\). Outside the null tie sets, exactly one index attains the maximum. For all sufficiently small \(t\), the same index continues to attain the maximum. Hence, \(D_{h,t}(p)\to\1_{G_h}(p)\) as \(t\to0\). The dominated convergence theorem, together with the term \(-w_hc_h\) in the definition of \(\Psi\), implies
\[\left.\frac{\mathrm d}{\mathrm dt} \Psi(c_1^*,\ldots,c_h^*+t,\ldots,c_s^*)\right|_{t=0} = \rho(G_h)-w_h.\]
Since \(c^*\) is a global minimizer of \(\Psi\), the function \(t\mapsto\Psi(c_1^*,\ldots,c_h^*+t,\ldots,c_s^*)\) is minimized at \(t=0\), and so its derivative at zero is zero. Hence, \(\rho(G_h)=w_h\). Since \(h\) was arbitrary, this equality holds for every \(h\). Each \(G_h\) is an intersection of affine halfspaces and has positive \(\rho\)-mass. If it had empty interior in \(H\), it would be contained in a proper affine hyperplane and would have zero \(\rho\)-mass. Thus, \((G_1,\ldots,G_s)\) is an \(s\)-cell power diagram. \end{proof}

Fix any \(c=(c_1,\ldots,c_s)\in\mathbb R^s\) such that the affine scores \(p\mapsto u_h(p)+c_h\), for \(h=1,\ldots,s\), generate an \(s\)-cell power diagram \((G_1,\ldots,G_s)\) satisfying \(\rho(G_h)=w_h\) for every \(h\); such a vector exists by \Cref{clm:prescribed-mass-diagram}. For \(h=1,\ldots,s\), let \(\rho^h\coloneqq\rho|_{G_h}/w_h\) denote the conditional law under \(\rho\) given \(G_h\).

\begin{claim}\label{clm:face-disintegrations}
A vector \(M\in\mathsf{MPC}_r(z,\rho)\) belongs to \(F_a\) if and only if there is a representation \eqref{eq:disintegration} satisfying
\[\sum_{i\in I_h}\frac{z_i}{w_h}\eta_i=\rho^h, \qquad\text{for }h=1,\ldots,s. \tag{5}\label{eq:local-disintegrations}\]
If \(M\in F_a\), then every representation \eqref{eq:disintegration} satisfies \eqref{eq:local-disintegrations}.
\end{claim}
\begin{proof}[Proof of \Cref{clm:face-disintegrations}]
Fix \(M\in\mathsf{MPC}_r(z,\rho)\), let \((q_i)_i\) satisfy \eqref{eq:density-rep}, and set \(\widehat q_h\coloneqq\sum_{i\in I_h}q_i\). Then
\[\Lambda_a(M)=\int_X\sum_{h=1}^s\widehat q_h(p)u_h(p)\,\rho(\mathrm dp)\leq\int_X\max_h\left\{u_h(p)+c_h\right\}\,\rho(\mathrm dp)-\sum_{h=1}^sw_hc_h. \tag{6}\label{eq:face-bound}\]
The inequality follows because \((\widehat q_h(p))_{h=1}^s\) is a probability vector for \(\rho\)-almost every \(p\), so \(\sum_{h=1}^s\widehat q_h(p)(u_h(p)+c_h)\leq\max_h\{u_h(p)+c_h\}\), while \(\int_X\widehat q_h(p)\,\rho(\mathrm dp)=w_h\). Up to the null cell boundaries, the functions \(p\mapsto(z_i/w_h)\1_{G_h}(p)\), for \(i\in I_h\), satisfy \eqref{eq:density-rep} and implement a vector that attains the bound. Outside the tie sets, there is a unique \(h\) such that \(u_h(p)+c_h=\max_g\{u_g(p)+c_g\}\). Hence, equality in \eqref{eq:face-bound} holds if and only if
\[\sum_{i\in I_h}q_i=\1_{G_h} \quad\rho\text{-almost everywhere}, \qquad\text{for }h=1,\ldots,s. \tag{7}\label{eq:face-equality}\]
If one representation of \(M\) satisfies \eqref{eq:face-equality}, then \(M\in F_a\). If \(M\in F_a\), every representation of \(M\) must attain the bound in \eqref{eq:face-bound} and, therefore, satisfy \eqref{eq:face-equality}. Since \(q_i(p)\rho(\mathrm dp)=z_i\eta_i(\mathrm dp)\), \eqref{eq:face-equality} is equivalent to \eqref{eq:local-disintegrations}.\end{proof}

\begin{claim}\label{clm:direct-face}
The coordinate-grouping map \(M\mapsto\left((m_i)_{i\in I_1},\ldots,(m_i)_{i\in I_s}\right)\) is an affine bijection from \(F_a\) onto
\[\prod_{h=1}^s \mathsf{MPC}_{\lvert I_h\rvert}(z^h,\rho^h). \tag{8}\label{eq:face-product}\]
\end{claim}
\begin{proof}[Proof of \Cref{clm:direct-face}]
By \Cref{clm:face-disintegrations} and \eqref{eq:disintegration}, if \(M\in F_a\), then \((m_i)_{i\in I_h}\in\mathsf{MPC}_{\lvert I_h\rvert}(z^h,\rho^h)\) for every \(h\).

Conversely, suppose that \((m_i)_{i\in I_h}\in\mathsf{MPC}_{\lvert I_h\rvert}(z^h,\rho^h)\) for every \(h\). By \eqref{eq:disintegration}, for each \(h\) there are probability measures \((\eta_i)_{i\in I_h}\) satisfying
\[\sum_{i\in I_h}\frac{z_i}{w_h}\eta_i=\rho^h \qquad \text{\&} \qquad m_{\eta_i}=m_i \quad\text{for every }i\in I_h.\]
Since \(w_h\rho^h=\rho|_{G_h}\) and the cells \(G_1,\ldots,G_s\) cover \(H\) up to \(\rho\)-null sets,
\[\sum_{i=1}^rz_i\eta_i = \sum_{h=1}^sw_h\rho^h = \rho.\]
Together with \(m_{\eta_i}=m_i\), \eqref{eq:disintegration} implies that \(M\in\mathsf{MPC}_r(z,\rho)\). The constructed representation satisfies \eqref{eq:local-disintegrations}, so \Cref{clm:face-disintegrations} implies that \(M\in F_a\). The coordinate-grouping map is onto. It is one-to-one because the grouped blocks contain every coordinate of \(M\), and it is affine because it only regroups those coordinates.\end{proof}

\begin{claim}\label{clm:recursive-extremality}
Each \(\rho^h\) is a regular reference law. Moreover, a vector \(M\in F_a\) is extreme in \(\mathsf{MPC}_r(z,\rho)\) if and only if \((m_i)_{i\in I_h}\) is extreme in \(\mathsf{MPC}_{\lvert I_h\rvert}(z^h,\rho^h)\) for every \(h\).
\end{claim}
\begin{proof}[Proof of \Cref{clm:recursive-extremality}]
Since \(\rho^h\ll\rho\),\footnote{This notation means that \(\rho^h\) is absolutely continuous with respect to \(\rho\).} every proper affine hyperplane is \(\rho^h\)-null. If \(\operatorname{aff}(\operatorname{supp}\rho^h)\) were a proper affine subspace of \(H\), it would be contained in a proper affine hyperplane to which \(\rho^h\) assigns probability one, a contradiction. Hence, \(\rho^h\) is regular.

Because \(F_a\) is a face, a vector in \(F_a\) is extreme in \(\mathsf{MPC}_r(z,\rho)\) if and only if it is extreme in \(F_a\). The affine bijection in \Cref{clm:direct-face} preserves extreme points, and a point of the Cartesian product in \eqref{eq:face-product} is extreme if and only if \((m_i)_{i\in I_h}\) is extreme in \(\mathsf{MPC}_{\lvert I_h\rvert}(z^h,\rho^h)\) for every \(h\).
\end{proof}

If \(s=1\), \eqref{eq:barycenter-constraint} makes \(\Lambda_a\) constant. If \(s\geq2\), it is nonconstant on the affine hull because the \(a_i\) are not all equal. Thus, every proper exposed face has the product representation in \eqref{eq:face-product}, with \(s\geq2\) and \(\lvert I_h\rvert<r\), and \Cref{clm:recursive-extremality} applies to it.

We prove \Cref{it:fixedweight} by induction on \(r\). For \(r=1\), \(\mathsf{MPC}_1((1),\rho)\) is a singleton. Let \(r\geq2\), and suppose the assertion holds for every smaller number of coordinates and every regular reference law.

Suppose first that \(M\) is extreme in \(\mathsf{MPC}_r(z,\rho)\). By \Cref{it:affine-hull}, \(\mathsf{MPC}_r(z,\rho)\) has dimension \((r-1)d>0\), so \(M\) lies in its relative boundary. The finite-dimensional supporting-hyperplane theorem places \(M\) in a proper exposed face \(F_a\). By \Cref{clm:recursive-extremality}, \((m_i)_{i\in I_h}\) is extreme in \(\mathsf{MPC}_{\lvert I_h\rvert}(z^h,\rho^h)\) for every \(h\). The induction hypothesis supplies, under \(\rho^h\), a hierarchical power diagram with cell masses \(z^h\) and barycenter vector \((m_i)_{i\in I_h}\). Refining each \(G_h\) by that diagram produces a hierarchical power diagram with cell masses \(z\) and barycenter vector \(M\).

Conversely, suppose that a hierarchical power diagram induces \(M\), and write \(K_i\) for its \(i\)th cell. Fix a sequence of refinements that constructs this diagram. Let \((G_1,\ldots,G_s)\) be the first power diagram in this sequence with at least two cells, and let \(I_h\) contain the indices of the final cells obtained by refining \(G_h\). Set \(w_h\coloneqq\rho(G_h)=\sum_{i\in I_h}z_i\), \(z^h\coloneqq(z_i/w_h)_{i\in I_h}\), and \(\rho^h\coloneqq\rho|_{G_h}/w_h\). Write the affine score generating \(G_h\) as \(p\mapsto\langle a^h,p-m_\rho\rangle+c_h\), and set \(a_i\coloneqq a^h\) for every \(i\in I_h\). \Cref{clm:recursive-extremality} implies that each \(\rho^h\) is a regular reference law.

Because \(\rho^h\) is concentrated on \(G_h\), the remaining refinements inside \(G_h\) induce, up to \(\rho^h\)-null sets, a hierarchical power diagram with probability vector \(z^h\) and barycenter vector \((m_i)_{i\in I_h}\). By the induction hypothesis, this block is extreme in \(\mathsf{MPC}_{\lvert I_h\rvert}(z^h,\rho^h)\). Set \(\eta_i\coloneqq\rho|_{K_i}/z_i\). Since the cells indexed by \(I_h\) partition \(G_h\) up to a \(\rho\)-null set,
\[\sum_{i\in I_h}\frac{z_i}{w_h}\eta_i=\frac{\rho|_{G_h}}{w_h}=\rho^h.\]
Hence, \Cref{clm:face-disintegrations} places \(M\) in the face \(F_a\). Since \((m_i)_{i\in I_h}\) is extreme in \(\mathsf{MPC}_{\lvert I_h\rvert}(z^h,\rho^h)\) for every \(h\), \Cref{clm:recursive-extremality} implies that \(M\) is extreme in \(\mathsf{MPC}_r(z,\rho)\).\end{proof}

\subsection{Proof of \texorpdfstring{\Cref{lem:fixed-probability-transfer}}{lemma}}
\begin{proof}[Proof of \Cref{lem:fixed-probability-transfer}]
By the overlap hypothesis, some set
\[A_n\coloneqq \left\{p\in X\colon q_i(p)\geq\frac{1}{n},\ q_j(p)\geq\frac{1}{n}\right\}\]
has positive \(\rho\)-probability. The restriction \(\rho|_{A_n}\) has full affine span in \(H\).\footnote{That is, \(\operatorname{aff}(\operatorname{supp}(\rho|_{A_n}))=H\).} Otherwise, its support would be contained in a proper affine hyperplane of \(H\) to which \(\rho\) assigns positive probability, contrary to regularity.

Let \(E\coloneqq H-H\), set \(m_n\coloneqq\rho(A_n)^{-1}\int_{A_n}p\,\rho(\mathrm dp)\), and define \(\Gamma_n\colon E\to E\) by
\[\Gamma_n a\coloneqq \int_{A_n} \left\langle a,p-m_n\right\rangle \left(p-m_n\right) \rho(\mathrm dp).\]
If \(\langle a,\Gamma_n a\rangle=0\), then \(p\mapsto\langle a,p-m_n\rangle\) vanishes on \(\operatorname{supp}(\rho|_{A_n})\). The full affine span of that support implies \(a=0\), so \(\Gamma_n\) is invertible.

For \(\xi\in E\), define \(f_\xi(p)\coloneqq \1_{A_n}(p) \left\langle\Gamma_n^{-1}\xi,p-m_n\right\rangle\); and so
\[\int_Xf_\xi(p) \rho(\mathrm dp)=0 \qquad \text{\&} \qquad \int_Xp f_\xi(p) \rho(\mathrm dp)=\xi.\]
The function \(f_\xi\) is bounded because \(X\) is compact. For \(\varepsilon>0\) sufficiently small, set \(q_i^\pm\coloneqq q_i\pm\varepsilon f_\xi\), \(q_j^\pm\coloneqq q_j\mp\varepsilon f_\xi\), and \(q_h^\pm\coloneqq q_h\) for \(h\notin\{i,j\}\). Since \(f_\xi\) vanishes outside \(A_n\) and \(q_i,q_j\geq1/n\) on \(A_n\), each \(q_h^\pm\) is nonnegative. The perturbed functions sum to one, preserve every message probability, and shift the first moments of messages \(i\) and \(j\) by \(\pm\varepsilon\xi\) and \(\mp\varepsilon\xi\), respectively. They, therefore, implement vectors \(M^+,M^-\in\mathsf{MPC}_r(z,\rho)\) with the coordinates in the statement and \(M=(M^++M^-)/2\).

Suppose \(m_i\neq m_j\), and set \(\xi=m_i-m_j\). Write
\[\overline m\coloneqq\frac{z_im_i+z_jm_j}{z_i+z_j}\qquad \text{\&} \qquad t\coloneqq1+\frac{\varepsilon(z_i+z_j)}{z_iz_j}>1.\]
Then \(m_i^+=\overline m+t(m_i-\overline m)\) and \(m_j^+=\overline m+t(m_j-\overline m)\). The following kernel is a martingale coupling from the normalized two-point law on \(m_i,m_j\) to the normalized two-point law on \(m_i^+,m_j^+\):
\[K_{hg}\coloneqq\frac{1}{t}\mathbf 1_{\{h=g\}}+ \left(1-\frac{1}{t}\right)\frac{z_g}{z_i+z_j}, \qquad \text{for } h,g\in\{i,j\}.\]
The kernel is stochastic, leaves the normalized weights \((z_i,z_j)/(z_i+z_j)\) invariant, and satisfies \(\sum_{g\in\{i,j\}}K_{hg}m_g^+=m_h\) for each \(h\in\{i,j\}\). Keeping every other atom fixed yields \(\nu_z(M)\preceq\nu_z(M^+)\). Evaluating the function \(p\mapsto\left\|p-\overline m\right\|^2\) reveals \(\nu_z(M)\prec\nu_z(M^+)\) because
\[z_i\left\|m_i^+-\overline m\right\|^2+z_j\left\|m_j^+-\overline m\right\|^2 =t^2\left(z_i\left\|m_i-\overline m\right\|^2+z_j\left\|m_j-\overline m\right\|^2\right),\]
which is strictly larger than \(z_i\left\|m_i-\overline m\right\|^2+z_j\left\|m_j-\overline m\right\|^2\).
\end{proof}

\subsection{Proof of \texorpdfstring{\Cref{lem:extreme-unique-disintegration}}{lemma}}

\begin{proof}[Proof of \Cref{lem:extreme-unique-disintegration}]
Existence follows from \eqref{eq:disintegration}. Let \(q\coloneqq(q_1,\ldots,q_r)\) satisfy \eqref{eq:density-rep}. If \(q(p)\) is not a vertex of \(\Delta_r\) on a set of positive \(\rho\)-probability, then some pair \(i\neq j\) satisfies \(q_i>0\) and \(q_j>0\) on a set of positive probability. Since \(d>0\), choose a nonzero \(\xi_0\in H-H\). Applying \Cref{lem:fixed-probability-transfer} with \(\xi=\xi_0\) produces distinct \(M^+,M^-\in\mathsf{MPC}_r(z,\rho)\) with \(M=(M^++M^-)/2\), contrary to extremality. Hence, every implementation of \(M\) is deterministic.

Suppose \(q\) and \(q'\) are two distinct implementations. Both are deterministic, while \((q+q')/2\) is an implementation that is not deterministic on the positive-probability set on which they differ, contradicting the fact that every implementation of \(M\) is deterministic. Thus, the implementation is unique. Writing \(D_i\coloneqq\{p\in X\colon q_i(p)=1\}\) produces the asserted partition and \(\eta_i=\rho|_{D_i}/z_i\). If a hierarchical power diagram \((K_i)_i\) induces \(M\), then \(q_i=\mathbf 1_{K_i}\) is an implementation, so uniqueness implies \(D_i=K_i\) \(\rho\)-almost everywhere.
\end{proof}

\subsection{Proof of \texorpdfstring{\Cref{lem:same-vocabulary-improvement}}{lemma}}
\begin{proof}[Proof of \Cref{lem:same-vocabulary-improvement}]
Since \(M\) is not extreme, there are distinct \(C,D\in\mathsf{MPC}_r(z,\rho)\) and \(t\in(0,1)\) such that \(M=tC+(1-t)D\). Fix \(0<\delta<\min\{t,1-t\}\), and set
\[M^0\coloneqq M-\delta(D-C) \qquad \text{\&} \qquad M^1\coloneqq M+\delta(D-C).\]
Then \(M^0,M^1\in\mathsf{MPC}_r(z,\rho)\), \(M^0\neq M^1\), and \(M=(M^0+M^1)/2\). Let \(q^0\) and \(q^1\) implement \(M^0\) and \(M^1\), and set \(q\coloneqq(q^0+q^1)/2\). Then \(q\) implements \(M\).

The rules \(q^0\) and \(q^1\) differ on a set of positive \(\rho\)-probability. On that set, \(q\) cannot be a vertex of \(\Delta_r\), because a vertex cannot be the average of two distinct points in \(\Delta_r\). Hence, some pair \(i\neq j\) satisfies \(q_i>0\) and \(q_j>0\) on a set of positive probability. Since the coordinates of \(M\) are distinct, \(m_i\neq m_j\). Apply \Cref{lem:fixed-probability-transfer} with \(\xi=m_i-m_j\). It yields vectors \(M'\) arbitrarily close to \(M\) satisfying \(\nu_z(M)\prec\nu_z(M')\). Sufficient closeness preserves distinctness of the coordinates.
\end{proof}

\subsection{Proof of \texorpdfstring{\Cref{lem:hierarchical-rigidity}}{lemma}}

\begin{proof}[Proof of \Cref{lem:hierarchical-rigidity}]
Write \(\lambda=\sum_{j=1}^sw_j\delta_{y_j}\), where the \(y_j\) are distinct and \(w_j>0\). By Strassen's theorem, there are numbers \(\pi_{ij}\geq0\) satisfying
\[\sum_{j=1}^s\pi_{ij}=z_i \quad \forall i, \qquad \sum_{i=1}^r\pi_{ij}=w_j \quad\forall j, \qquad \text{\&} \qquad \sum_{j=1}^s\pi_{ij}y_j=z_im_i \quad\forall i.\]

Applying \eqref{eq:disintegration} to \(\lambda\preceq\mu\), choose probability measures \(\theta_1,\ldots,\theta_s\) satisfying
\[\sum_{j=1}^sw_j\theta_j=\mu \qquad \text{\&} \qquad m_{\theta_j}=y_j \quad\text{for every }j.\]

For each \(i\), define \(\eta_i\coloneqq\frac{1}{z_i}\sum_{j=1}^s\pi_{ij}\theta_j\), and because \(\sum_{j=1}^s\pi_{ij}=z_i\), each \(\eta_i\) is a probability measure. Moreover,
\[\sum_{i=1}^rz_i\eta_i =\sum_{i=1}^r\sum_{j=1}^s\pi_{ij}\theta_j =\sum_{j=1}^s\left(\sum_{i=1}^r\pi_{ij}\right)\theta_j =\sum_{j=1}^sw_j\theta_j =\mu.\]
For every \(i\),
\[m_{\eta_i} =\frac{1}{z_i}\sum_{j=1}^s\pi_{ij}m_{\theta_j} =\frac{1}{z_i}\sum_{j=1}^s\pi_{ij}y_j =m_i.\]
Hence, \(\eta_1,\ldots,\eta_r\) satisfy \eqref{eq:disintegration}.

By \Cref{it:fixedweight} of \Cref{prop:fixedweight}, the barycenter vector of \(\nu\) is extreme. Hence, \Cref{lem:extreme-unique-disintegration} implies \(\eta_i=\frac{\mu|_{K_i}}{z_i}\) for every \(i\).

Suppose \(\pi_{ij}>0\). Since \(\eta_i(H\setminus K_i)=0\),
\[0=\eta_i(H\setminus K_i) \geq\frac{\pi_{ij}}{z_i}\theta_j(H\setminus K_i),\]
and so \(\theta_j(K_i)=1\). The same index \(j\) cannot satisfy both \(\pi_{ij}>0\) and \(\pi_{i'j}>0\) for distinct \(i\) and \(i'\). Indeed, those inequalities would imply \(\theta_j(K_i\cap K_{i'})=1\). Since \(\sum_{j=1}^sw_j\theta_j=\mu\) and \(w_j>0\), we have \(\theta_j\ll\mu\). The cells \(K_i\) and \(K_{i'}\) are disjoint up to a \(\mu\)-null set, so \(\theta_j(K_i\cap K_{i'})=0\), contradicting \(\theta_j(K_i\cap K_{i'})=1\).

Every \(i\) must be paired with at least one \(j\), because \(\sum_j\pi_{ij}=z_i>0\), and each \(j\) can be paired with at most one \(i\). Consequently, \(s\geq r\). If \(s=r\), every \(i\) is paired with exactly one \(j\), and every \(j\) is paired with exactly one \(i\). After relabeling, the coupling identities and the equality \(\eta_i=\mu|_{K_i}/z_i\) deliver
\[w_i=z_i, \qquad y_i=m_i, \qquad \text{\&} \qquad \theta_i=\frac{\mu|_{K_i}}{z_i}.\]
Thus, \(\lambda=\nu\), contrary to the hypothesis. Therefore, \(s\geq r+1\).
\end{proof}

\subsection{Proof of \texorpdfstring{\Cref{lem:one-more-message}}{lemma}}

\begin{proof}[Proof of \Cref{lem:one-more-message}]
Write \(\nu=\sum_{i=1}^rz_i\delta_{m_i}\). By \eqref{eq:disintegration}, choose probability measures \(\theta_1,\ldots,\theta_r\) satisfying
\[\sum_{i=1}^rz_i\theta_i=\mu \qquad \text{\&} \qquad m_{\theta_i}=m_i \quad\text{for every }i.\]
Since every \(z_i>0\), the first equality implies \(\theta_i\ll\mu\) for every \(i\). Hence, \Cref{ass:regular-posterior} implies that every proper affine hyperplane is \(\theta_i\)-null. Fix an index \(i\), and choose a nonconstant affine function \(a\colon H\to\mathbb R\). Every level set of \(a\) is a proper affine hyperplane, so the distribution of \(a\) under \(\theta_i\) is atomless. It is also nondegenerate, because otherwise \(\theta_i\) would be supported on one such hyperplane. For \(c\in\mathbb R\), set
\[A_c^+\coloneqq\{p\in H\colon a(p)>c\} \qquad \text{\&} \qquad A_c^-\coloneqq\{p\in H\colon a(p)<c\}.\]
Choose \(c\) so that both sets have positive \(\theta_i\)-probability, and set
\[q\coloneqq\theta_i(A_c^+), \qquad m_i^+\coloneqq\frac{1}{q}\int_{A_c^+}p\,\theta_i(\mathrm dp), \qquad \text{\&} \qquad m_i^-\coloneqq\frac{1}{1-q}\int_{A_c^-}p\,\theta_i(\mathrm dp).\]
Since \(\theta_i(\{p\colon a(p)=c\})=0\) and \(a\) is affine,
\[m_i=qm_i^++(1-q)m_i^- \qquad \text{\&} \qquad a(m_i^+)>c>a(m_i^-).\]
In particular, \(m_i^+\neq m_i^-\). Define
\[\lambda\coloneqq \sum_{h\neq i}z_h\delta_{m_h} +z_iq\delta_{m_i^+} +z_i(1-q)\delta_{m_i^-}.\]
Replacing \(m_i\) by a nondegenerate two-point distribution with mean \(m_i\) yields \(\nu\preceq\lambda\). The comparison is strict because
\[q\left\|m_i^+\right\|^2 +(1-q)\left\|m_i^-\right\|^2 -\left\|m_i\right\|^2 =q(1-q)\left\|m_i^+-m_i^-\right\|^2 >0.\]
Splitting \(\theta_i\) into its conditional laws on \(A_c^+\) and \(A_c^-\) delivers a martingale disintegration of \(\mu\) over \(\lambda\), so \(\lambda\preceq\mu\). The law \(\lambda\) has at most \(r+1\) support points.

If \(\nu\) is induced by a hierarchical power diagram \((K_1,\ldots,K_r)\), \Cref{it:fixedweight} of \Cref{prop:fixedweight} implies that its barycenter vector is extreme. Hence, \Cref{lem:extreme-unique-disintegration} implies \(\theta_i=\mu|_{K_i}/z_i\). Refining \(K_i\) by the two-cell power diagram determined by \(a(p)=c\) produces a hierarchical power diagram inducing \(\lambda\). Finally, \Cref{lem:hierarchical-rigidity} implies that a strict improvement has at least \(r+1\) support points. Thus, \(\lambda\) has exactly \(r+1\) support points and is an HPD summary.
\end{proof}

\end{document}